\documentclass[
 reprint,
 amsmath,amssymb,
 aps,nofootinbib]{revtex4-2}

\usepackage{graphicx}
\usepackage{subcaption}
\usepackage{dcolumn}
\usepackage{bm}
\usepackage{hyperref}
\usepackage{braket}
\usepackage{xcolor}
\usepackage{comment}
\usepackage{amsthm}
\usepackage{tikz}
\usetikzlibrary{decorations.pathmorphing,positioning,arrows.meta,calc}
\usepackage{algorithm}
\usepackage{algpseudocode}
\newtheorem{theorem}{Theorem}
\newtheorem{lemma}{Lemma}
\newtheorem{remark}{Remark}
\newtheorem{proposition}{Proposition}
\newtheorem{corollary}{Corollary}
\newtheorem{problem}{Problem}
\newtheorem{definition}{Definition}

\newcommand{\francisco}[1]{}

\begin{document}

\preprint{APS/123-QED}

\title{Entanglement improves coordination in distributed systems}

\author{Francisco Ferreira da Silva}
    \thanks{Corresponding author: \href{mailto:francisco@delftnetworks.com}{francisco@delftnetworks.com}}
    \affiliation{Delft Networks, Lorentzweg 1, 2628 CJ Delft, The Netherlands}
    \affiliation{QuTech, Delft University of Technology, Lorentzweg 1, 2628 CJ Delft, The Netherlands}
\author{Stephanie Wehner}
    \affiliation{Delft Networks, Lorentzweg 1, 2628 CJ Delft, The Netherlands}
    \affiliation{QuTech, Delft University of Technology, Lorentzweg 1, 2628 CJ Delft, The Netherlands}
    \affiliation{Quantum Computer Science, EEMCS, Delft University of Technology, Lorentzweg 1, 2628 CJ Delft, The Netherlands}
    \affiliation{Kavli Institute of Nanoscience, Delft University of Technology, Lorentzweg 1, 2628 CJ Delft, The Netherlands}
\begin{abstract}
Coordination in distributed systems is often hampered by communication latency, which degrades performance.
Quantum entanglement offers fundamentally stronger correlations than classically achievable without communication.
Crucially, these correlations manifest instantaneously upon measurement, irrespective of the physical distance separating the systems.
We investigate the application of shared entanglement to a dual-work optimization problem in a distributed system comprising two servers.
The system must process both a continuously available, preemptible baseline task and incoming customer requests arriving in pairs.
System performance is characterized by the trade-off between baseline task throughput and customer waiting time.
We present a rigorous analytical model demonstrating that when the baseline task throughput function is strictly convex, rewarding longer uninterrupted processing periods, entanglement-assisted routing strategies achieve Pareto-superior performance compared to optimal communication-free classical strategies.
We prove this advantage through queueing-theoretic analysis, non-local game formulation, and computational certification of classical bounds.
We complement these formal results with simulations indicating that the advantage persists when arrivals are independent across routers, and grows under bursty traffic.
Our results suggest distributed scheduling and coordination as a candidate application domain for near-term entanglement-based quantum networks.
\end{abstract}
\maketitle
\section{Introduction}
Coordination enables efficient operation in distributed systems.
An important application lies in scheduling, where incoming requests must be assigned across multiple servers~\cite{casavant1988taxonomy}.
Optimal scheduling and load balancing often rely on global state information, such as current server loads or queue lengths.
However, acquiring this information via classical communication introduces latency.
In latency-sensitive scenarios, this delay can render state information obsolete, leading to suboptimal decisions based on outdated data and consequently degrading overall system performance~\cite{wang1985load,jiang2015survey}.
For instance, routing incoming user requests without real-time knowledge of server availability can cause load imbalance and increase user wait times.
More generally, when routing or scheduling decisions must be made on timescales that are short compared to the communication delay between routers and servers, any attempt to gather fresh state before each decision either introduces unacceptable delay or relies on information that is already stale by the time it is used.
For example, in a wide-area deployment with inter-site separations of order $100$ km, classical round-trip latencies are typically in the sub-millisecond to millisecond range, while local routing or scheduling decisions inside a high-speed service may need to be taken on timescales of tens of microseconds or less.
In such regimes, global coordination based on classical communication becomes fundamentally limited.

Entanglement offers a fundamentally new approach to coordination.
It provides a mechanism for establishing correlations between spatially separate systems that are stronger than any achievable classically without communication~\cite{bell1964einstein,brunner2014bell}.
This mechanism can be implemented as follows.
Initially, the coordinating parties share an entangled quantum state.
At a later time, upon receiving local information relevant to their coordination task, each party performs a measurement on their component of the entangled state, which can be conditioned on the local information they received.
The outcomes of these local measurements will exhibit strong non-local correlations and can be used to guide the parties' decisions, thus enabling coordination without communication.

In this work, we investigate the application of entanglement-assisted coordination to a routing problem in a distributed system with two servers.
The system must process both a continuously available, preemptible baseline task and incoming customer requests.
System performance is characterized by the trade-off between baseline task throughput and customer waiting time.
The challenge lies in coordinating the assignment of incoming requests to the servers based only on local information, namely the processing time required by the local request, under latency constraints that preclude effective real-time communication.
In this work we therefore take strictly non-communicating classical routing policies, where each server's decision depends only on its own local information, as the baseline for comparison.
We emphasize that this restriction is imposed symmetrically on quantum and classical strategies, so that any performance gap isolates the contribution of entanglement at a fixed (zero) communication budget; comparing entanglement-assisted strategies against classical protocols that use delayed or coarse-grained state information is a separate question, which we leave to future work.

We show that when the cumulative baseline output function $T(t)$ is strictly convex, entanglement-assisted routing achieves Pareto-superior performance compared to optimal classical strategies without communication.
Heralded entanglement generation between physically separated systems has been demonstrated in multiple qubit platforms~\cite{covey2023quantum,ruf2021quantum,krutyanskiy2023entanglement}, including over deployed fiber~\cite{stolk2024metropolitan}.
This makes entanglement-assisted coordination an attractive near-term application of quantum networks.

\section{Related Work}
This work provides a complete analytical treatment of the routing problem introduced in~\cite{da2025entanglement}, where we first demonstrated quantum advantages in distributed scheduling.
The present manuscript extends that work with full queueing-theoretic proofs, a rigorous mapping to a weighted non-local game, certified classical bounds and numerical robustness checks against the simplifying assumptions on arrival processes.

The underlying principle behind leveraging entanglement for coordination traces back to Bell's theorem, which established that quantum mechanics predicts correlations stronger than any classical theory permits~\cite{bell1964einstein,hensen2015loophole}.
This phenomenon is formalized through non-local games, where non-communicating players cooperate to maximize a payoff; for some games, quantum strategies outperform classical ones~\cite{clauser1969proposed,brunner2014bell}.
Related but orthogonal to our focus,~\cite{ding2025quantum} extends the non-local game framework to settings where parties can communicate subject to timing constraints; our model corresponds to the zero-communication extreme.

Several works have explored translating quantum advantages in abstract non-local games into practical benefits by mapping coordination problems onto game structures.
Examples include market making in high-frequency trading~\cite{ding2024coordinating}, load balancing in ad-hoc networks~\cite{hasanpour2017quantum}, rendezvous tasks~\cite{mironowicz2023entangled,viola2024quantum,tucker2024quantum}, and broader networked systems~\cite{arun2025faster}.

From a different perspective, the problem studied in this work falls within the classical domains of load balancing, scheduling theory, and queueing theory; for foundational concepts, see~\cite{wang1985load,casavant1988taxonomy,gross2011fundamentals}.
Our work differs from classical approaches by introducing entanglement as a coordination mechanism.

Classical randomized load-balancing schemes, such as join-the-shortest-queue and the power-of-two-choices family, can dramatically improve performance compared to naive routing by using a small amount of communication to obtain partial state information~\cite{mitzenmacher2002power,gupta2007analysis}.
Our setting is different in that we explicitly target regimes where the routing decision must be made on timescales short compared to the round-trip communication delay between servers, so even this would either introduce unacceptable delay or rely on stale information.
Accordingly, we take strictly non-communicating classical policies as our baseline, and compare them to entanglement-assisted strategies that achieve stronger non-local correlations without real-time communication.


\section{System Description}
\label{sec:system_description}

We now introduce the distributed system we consider, depicted in Figure~\ref{fig:system}.
It consists of two identical servers that process work at rate $\mu$.
Each server maintains a queue of unlimited capacity and follows a first-come, first-served (FCFS) discipline.

\begin{figure}[!htpb]
\centering
\includegraphics[width=\columnwidth]{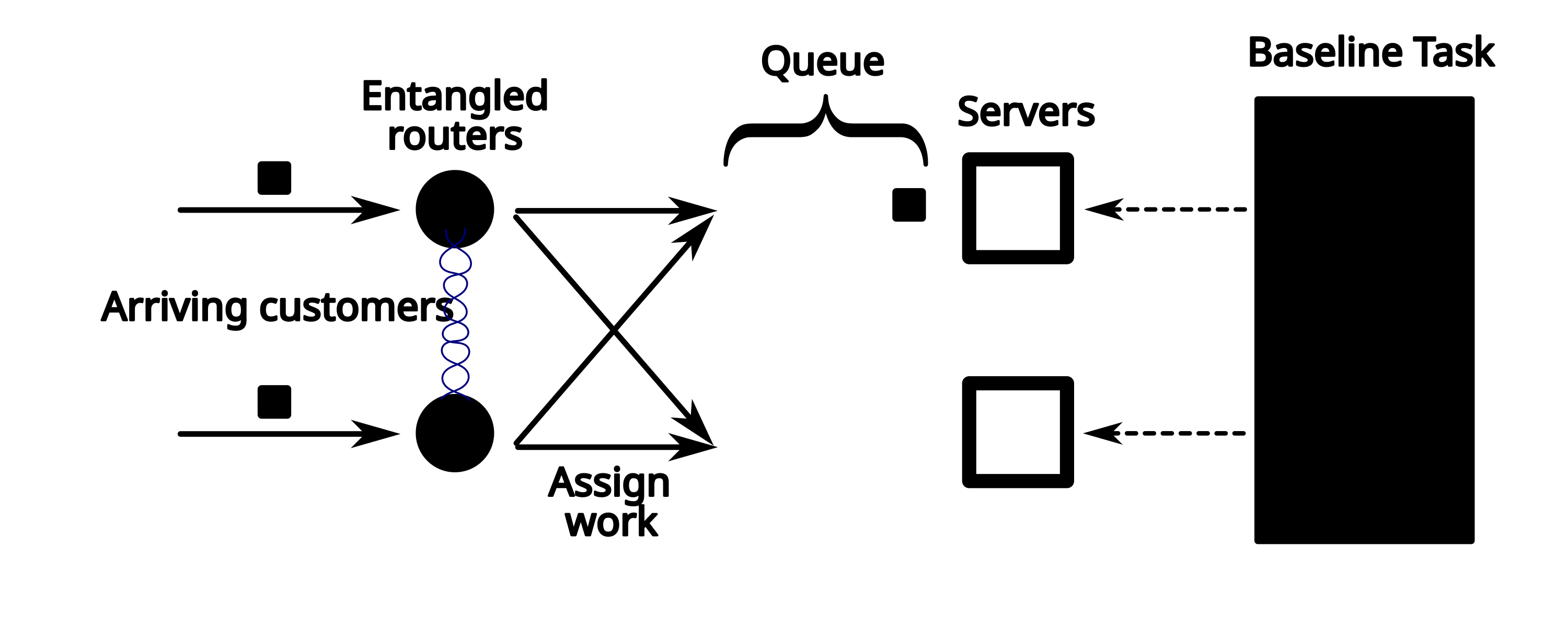}
\caption{Distributed system studied in this work.
Two servers, depicted by hollow squares, handle two distinct types of work.
On the right, a baseline task that is preemptible and always available.
On the left, customer requests that dynamically arrive at the system via the routers, depicted as circles.
The servers are endowed with queues of unlimited size in which customers wait.
The routers may share entanglement with each other, which they can use to better coordinate their routing decisions.
Entanglement is depicted here by wavy lines connecting the two routers.}
\label{fig:system}
\end{figure}

The system handles two distinct types of work.
First, a continuously available baseline task that is always present and can be processed by either server.
This baseline task is preemptible: a server working on it immediately switches to servicing a customer upon assignment.
Servers are never idle, as when a server's queue is empty, it resumes processing the baseline task.

Second, customer requests that arrive dynamically in pairs.
Specifically, we assume that pairs of customers arrive simultaneously, with one customer arriving at each of two routers (depicted as circles in Figure~\ref{fig:system}).
Customer pairs arrive according to a Poisson process with rate $\lambda$.
The service time $X_i$ required by customer $i$ is drawn independently from an exponential distribution $\text{Exp}(\mu)$ with mean $1/\mu$.
This yields a two-server FCFS queueing system with Poisson pair arrivals and exponential service times, with per-server utilization $\rho = \lambda/\mu < 1$~\cite{gross2011fundamentals}.

Upon arrival of a customer pair $(X_1, X_2)$, the routers must assign their respective customers to one of the two servers.
Each router observes only the service time of its local customer.
A key constraint is that the routers cannot communicate in real time to coordinate this assignment.
This models scenarios where physical distance introduces communication latencies that are significant compared to the decision timescale, so exchanging useful real-time state is infeasible.

Each routing strategy induces a splitting probability $p$, the long-run fraction of customer pairs sent to different servers.
When customers are bunched (probability $1-p$), one is selected uniformly at random to precede the other in the queue.

We evaluate system performance along two dimensions.
First, $W_q$, the average time customers spend waiting in queue before service begins.

Second, baseline throughput, which measures the system's productivity on the continuously available baseline task.
Recall that servers work on this baseline task whenever their queue is empty, switching to customer service when assignments arrive.
Let $T(t)$ denote the output produced when a server works uninterrupted on the baseline task for duration $t$.
We assume $T$ is differentiable, increasing, with $T(0) = 0$.
The long-run average baseline throughput per server, which we denote $\mathcal{T}$, quantifies how much baseline work the system completes over time.

\section{The Coordination Challenge}
\label{sec:coordination}
We show that the optimal routing policy for the routing problem we defined requires both routers to know both service times $(X_1, X_2)$, making it unattainable for routers restricted to local observations and pre-shared resources.
This establishes that the routers are faced with a coordination challenge.

A rigorous queueing-theoretic analysis of the model introduced in the previous section, including proofs for statements made in this section, is given in Appendix~\ref{sec:delta_threshold}.
We start by noting that $\mathcal{T}$ depends only on the overall splitting probability $p$, not on which specific pairs are split.
This is because throughput depends only on the frequency and duration of idle periods, which are determined by the splitting probability.

In contrast, customer waiting time depends on which pairs are split.
Splitting pairs with high service times is disproportionately helpful in reducing waiting time.
This is due to the dependence of the waiting time on the second moment of the service time distribution~\cite{gross2011fundamentals}.

Given that baseline throughput depends only on $p$, all routing strategies resulting in the same splitting probability will achieve the same throughput.
Hence, among this group of strategies, we wish to find the one that minimizes waiting time.
Finding the optimal routing strategy for a fixed splitting probability $p$ is equivalent to Problem~\ref{prob:fixed_p_main}.
By solving this for all $p \in [0,1]$, we can characterize the optimal trade-off between the two objectives.

\begin{problem}[Optimal routing at fixed splitting probability]\label{prob:fixed_p_main}
Given a target splitting probability $p \in [0,1]$, find a routing policy $r: \mathbb{R}_+^2 \to [0,1]$ that solves
\begin{equation}
\begin{aligned}
\text{maximize} \quad &\mathbb{E}[r(X_1,X_2) \cdot w(X_1, X_2)] \\
\text{subject to} \quad & \mathbb{E}[r(X_1,X_2)] = p, \\
& 0 \leq r(x_1,x_2) \leq 1 \text{ for all } (x_1,x_2).
\end{aligned}
\end{equation}
\end{problem}
Here $r(x_1, x_2)$ denotes the probability of splitting a specific pair with realized service times $(x_1, x_2)$, under load-balanced server assignment: when splitting, each server is equally likely to receive either customer; when bunching, the destination server is chosen uniformly at random.\footnote{Throughout this manuscript, we use uppercase letters (e.g., $X_1$) to denote random variables and lowercase letters (e.g., $x_1$) to denote their realizations or dummy variables in function definitions.}
This restriction is without loss of generality, as load-balanced policies achieve weakly lower waiting time at the same throughput (Appendix~\ref{sec:symmetrization}).
The splitting benefit $w(x_1,x_2)$ is given by
\begin{equation}
w(x_1, x_2) = c_1 x_1 x_2 + c_2(x_1 + x_2),
\end{equation}
where $c_1 = \frac{\lambda}{2(1-\rho)}$ and $c_2 = \frac{1}{4}$.
This function quantifies the waiting time reduction from splitting rather than bunching a pair with service times $(x_1, x_2)$; higher values indicate pairs where splitting provides greater benefit.
The first term arises from the Pollaczek--Khinchine formula~\cite{gross2011fundamentals} for the $M^X/G/1$ queue, and captures the queueing benefit of reducing workload variance by sending the two customers to different servers.
The second captures the within-batch delay: when a pair is bunched behind a single server, one customer waits for the other to finish service.
A full derivation is given in Appendix~\ref{sec:delta_threshold}.
We show in Appendix~\ref{sec:optimal_policy} that the solution to this optimization problem has a simple structure: it is a threshold policy that splits pairs when $w(X_1, X_2) > \tau_p$ and bunches pairs when $w(X_1, X_2) < \tau_p$, where threshold $\tau_p$ is chosen such that $\Pr[w(X_1, X_2) > \tau_p] = p$.
We call this the $w$-threshold policy.

\begin{definition}[Pareto Optimality]
A routing policy is Pareto optimal with respect to customer waiting time $W_q$ and baseline throughput $\mathcal{T}$ if there exists no other feasible policy that can strictly improve one objective (decrease $W_q$ or increase $\mathcal{T}$) without worsening the other. The set of all such policies constitutes the Pareto frontier.
\end{definition}

\begin{theorem}[Pareto Optimality]
\label{thm:pareto}
Assume the baseline throughput function $T(t)$ is strictly convex.
Then the family of $w$-threshold policies, parametrized by $p \in [0,1]$, traces out the complete Pareto frontier between customer waiting time and baseline throughput.
\end{theorem}

The strict convexity assumption on $T(t)$ captures scenarios where longer uninterrupted processing periods yield disproportionately more output, for example, tasks with setup costs, learning curves, or context-switching penalties.
Under this assumption, both baseline throughput and customer waiting time strictly decrease with $p$, creating a trade-off between the two objectives.

The $w$-threshold policy requires evaluating $w(X_1, X_2)$, which depends on both service times.
Without real-time communication, however, each router observes only its local service time: router $A$ sees $X_1$, and router $B$ sees $X_2$.
Neither router can directly compute $w(X_1, X_2)$.

This establishes the coordination challenge: optimal routing decisions are determined by a global quantity $w(X_1, X_2)$, but routers are constrained to make decisions based on local observations and any pre-shared resources.

In the quantum information literature, such coordination problems are naturally formulated as non-local games, where non-communicating players receive inputs and must produce correlated outputs to maximize a payoff.
Bell's theorem established that quantum entanglement enables stronger correlations than classically achievable without communication, suggesting entanglement might help routers better approximate the $w$-threshold policy.

\section{Entanglement as a Coordination Resource}
\label{sec:entanglement}

Having established that optimal routing requires non-local information $w(X_1, X_2)$, we now formulate the coordination problem as a non-local game and show how entanglement enables better approximation of the $w$-threshold policy.

A non-local game is a cooperative game between spatially separated players who cannot communicate~\cite{brunner2014bell}.
Each player receives an input, produces an output, and the players' joint performance is evaluated by a payoff function depending on all inputs and outputs.
Formally:
\begin{definition}[Non-local game]
\label{def:nonlocal-game}
A non-local game $G = (\mathcal{X}, \mathcal{Y}, \mathcal{A}, \mathcal{B}, \pi, V)$ consists of input sets $\mathcal{X}, \mathcal{Y}$, output sets $\mathcal{A}, \mathcal{B}$, a probability distribution $\pi$ over $\mathcal{X} \times \mathcal{Y}$, and a payoff function $V: \mathcal{X} \times \mathcal{Y} \times \mathcal{A} \times \mathcal{B} \to \mathbb{R}$.
The game proceeds as follows: inputs $(x, y) \sim \pi$ are distributed to players A and B respectively; each player produces an output ($a \in \mathcal{A}$, $b \in \mathcal{B}$) without communication; the payoff $V(x, y, a, b)$ is evaluated.
Players may share pre-distributed resources (classical shared randomness or quantum entanglement) but cannot communicate during the game.
\end{definition}

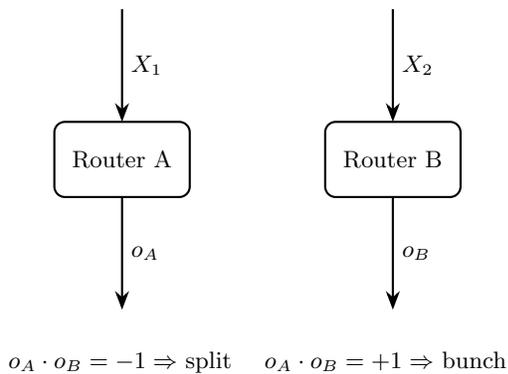
\begin{figure}[t]
\centering
\begin{tikzpicture}[
    router/.style={draw, rounded corners, minimum width=1.8cm, minimum height=1.0cm, thick},
    arrow/.style={-{Stealth[length=2.5mm]}, thick}
]

\node[router] (A) at (-1.8, 0) {Router A};
\node[router] (B) at (1.8, 0) {Router B};

\draw[arrow] (-1.8, 2.0) -- (-1.8, 0.5) node[midway, right] {$X_1$};
\draw[arrow] (1.8, 2.0) -- (1.8, 0.5) node[midway, right] {$X_2$};

\draw[arrow] (-1.8, -0.5) -- (-1.8, -2.0) node[midway, right] {$o_A$};
\draw[arrow] (1.8, -0.5) -- (1.8, -2.0) node[midway, right] {$o_B$};

\node[align=center] at (0, -2.7) {\small $o_A \cdot o_B = -1 \Rightarrow$ split \quad $o_A \cdot o_B = +1 \Rightarrow$ bunch};

\end{tikzpicture}
\caption{The routing problem as a non-local game.
Routers A and B receive inputs (service times $X_1$, $X_2$) and produce outputs ($o_A, o_B \in \{+1,-1\}$) without communication.
The product $o_A \cdot o_B$ determines whether the customer pair is split across servers or bunched to the same server.}
\label{fig:nonlocal-game}
\end{figure}

We map the routing problem to a non-local game (Figure~\ref{fig:nonlocal-game}) in which the routers act as players A and B.
The inputs to the players are the observed service times $X_1$ and $X_2$.
Each player's output specifies which server to route their customer to, labelled $\pm 1$.
Note that the splitting probability from Section~\ref{sec:coordination} relates to the outputs by $r(x_1, x_2) = \Pr[o_A(x_1) \cdot o_B(x_2) = -1]$.

Unlike standard non-local games where performance is measured by winning probability, we weight decisions by $w(X_1, X_2)$, reflecting the disproportionate impact of misrouting high-service-time pairs on waiting time~\cite{gross2011fundamentals}.
The expected payoff under strategies $(o_A, o_B)$ is $\mathbb{E}_{(X_1, X_2) \sim \pi}[V(X_1, X_2, o_A(X_1), o_B(X_2))]$; we define this as
\begin{equation}
A(o_A, o_B) = -\mathbb{E}[o_A(X_1) o_B(X_2) \, w(X_1, X_2)],
\end{equation}
where $o_A: \mathbb{R}_+ \to \{+1, -1\}$ and $o_B: \mathbb{R}_+ \to \{+1, -1\}$ denote the decision functions mapping observed service times to output choices.
Since splitting corresponds to $o_A o_B = -1$, higher $A$ rewards splitting high-$w$ pairs, precisely what the optimal policy does, translating directly into reduced excess waiting time.

The $w$-threshold policy corresponds to $o_A \cdot o_B = \sigma^*(X_1, X_2)$, where $\sigma^*(x_1, x_2) = \mathrm{sign}(\tau_p - w(x_1, x_2))$.

\begin{lemma}[Payoff $\leftrightarrow$ Waiting Time Gap]
\label{lem:agreement_queue_main}
For any routing strategy with game payoff $A$, the excess customer waiting time relative to the $w$-threshold policy satisfies
\begin{equation}
\Delta W_q = \frac{A^* - A}{2},
\end{equation}
where $A^* = -\mathbb{E}[\sigma^*(X_1, X_2) \, w(X_1, X_2)]$ is the $w$-threshold policy's payoff.
\end{lemma}

The proof is given in Appendix~\ref{sec:game-mapping}.
Since all strategies at the same splitting probability $p$ achieve identical baseline throughput, a strategy with higher $A$ provides a Pareto improvement in the waiting time--throughput trade-off.

We now describe how entanglement serves as a coordination resource.
The routers pre-share an entangled state, e.g., $\ket{\psi^-} = (\ket{01} - \ket{10})/\sqrt{2}$.
Upon observing local service times $X_1$ and $X_2$, router A performs a measurement on its qubit in a basis determined by angle $\theta_A(X_1)$, obtaining outcome $o_A \in \{+1, -1\}$; router B acts analogously with angle $\theta_B(X_2)$.
The routers then route to server $o_A$ and $o_B$ respectively, resulting in bunching when $o_A \cdot o_B = +1$ and splitting when $o_A \cdot o_B = -1$.

The quantum advantage question is now precise: can entanglement-assisted strategies achieve $A_{\text{quantum}} > A_{\text{classical}}$ for this weighted non-local game?
In the next section, we establish this affirmatively for a range of system parameters.

\section{Main Results}
\label{sec:main_results}

We now present the main results of this work.
First, we establish that quantum advantage in the routing game implies superior routing performance.
Second, we prove that optimal classical strategies take a simple threshold form.
Third, we demonstrate that quantum advantage manifests across a substantial range of system parameters.

\begin{theorem}[Quantum Advantage in Routing]
\label{thm:main}
Assume the baseline throughput function $T(t)$ is strictly convex.
Let $\mathcal{P} = \{p \in [0,1] : A_{\mathrm{qu}}(p) > A_{\mathrm{cl,SR}}^*(p)\}$ denote the set of splitting probabilities at which quantum strategies achieve higher game payoff than classical strategies with shared randomness.
Then quantum strategies Pareto-dominate classical strategies over $\mathcal{P}$: for all $p \in \mathcal{P}$,
\begin{equation}
\Delta W_q^{\mathrm{qu}}(p) < \Delta W_q^{\mathrm{cl}}(p)
\end{equation}
at identical baseline throughput $\mathcal{T}(p)$.
\end{theorem}

\begin{proof}
Under strict convexity, baseline throughput $\mathcal{T}(p)$ is strictly monotonic in $p$ (Proposition~\ref{prop:baseline-monotonicity}).
Equal throughput therefore implies equal splitting probability.
By Lemma~\ref{lem:agreement_queue_main}, $\Delta W_q = (A^* - A)/2$.
Hence $A_{\mathrm{qu}} > A_{\mathrm{cl,SR}}^*$ implies $\Delta W_q^{\mathrm{qu}} < \Delta W_q^{\mathrm{cl}}$ at the same $p$, and thus at the same throughput $\mathcal{T}(p)$.
\end{proof}

Theorem~\ref{thm:main} reduces the question of whether entanglement improves routing performance to determining the advantage region $\mathcal{P}$.
To establish that $\mathcal{P}$ is nonempty, we need lower bounds on quantum performance and upper bounds on classical performance.
Lower bounds are straightforward: we construct explicit quantum strategies and evaluate their payoffs.
Upper bounds are in principle more challenging, since we must optimize over all possible local decision functions, i.e., an infinite-dimensional space.

The following theorem provides the structural result that makes classical certification tractable.

\begin{theorem}[Classical Strategies are Threshold Strategies]
\label{thm:classical-threshold-main}
For the routing game, any optimal deterministic classical strategy $(o_A^*, o_B^*)$ consists of threshold functions: there exist $\theta_A, \theta_B \geq 0$ such that
\begin{equation}
o_A^*(x) = \begin{cases} +1 & x < \theta_A \\ -1 & x \geq \theta_A \end{cases}, \qquad
o_B^*(x) = \begin{cases} +1 & x < \theta_B \\ -1 & x \geq \theta_B \end{cases}.
\end{equation}
\end{theorem}

The proof is given in Appendix~\ref{app:classical_strategies}.
This result reduces classical certification to a tractable finite-dimensional optimization over threshold pairs $(\theta_A, \theta_B)$.
Classical players may also use shared randomness to correlate their strategies; the optimal shared-randomness payoff is $A_{\mathrm{cl,SR}}^*(p) = \mathrm{conc}(A_{\mathrm{cl}}^*)(p)$, the concave envelope of the deterministic value.
Our certified bounds compare against this benchmark; see Appendix~\ref{app:classical_strategies} for details.

\begin{theorem}[Quantum Advantage Region]
\label{thm:advantage-region}
For exponential service times with $\mu = 1$ and arrival rate $\lambda = 0.8$, quantum strategies achieve certified advantage over classical strategies with shared randomness for splitting probabilities $p \in [0.075, 0.325]$.
The maximum reduction in waiting time gap ($\Delta W_q^{\mathrm{cl}} - \Delta W_q^{\mathrm{qu}}$) is approximately $0.073$ (in units of $1/\mu$), occurring near $p \approx 0.20$.
\end{theorem}

The specific values reported here correspond to a stylized parameter regime; they should be interpreted as a certified existence proof of quantum advantage rather than as a quantitative prediction for any deployed system.
The dependence of the magnitude of the advantage on the arrival process is explored numerically in Section~\ref{sec:independent_numerical}.

Both the quantum and classical bounds are obtained numerically.
Quantum lower bounds come from explicit strategies parametrised by polynomial measurement angles, optimised by a Gauss-Laguerre quadrature scheme.
Classical upper bounds, on the other hand, exploit Theorem~\ref{thm:classical-threshold-main} to reduce the certification problem to a grid search over threshold pairs $(\theta_A, \theta_B)$, with the residual gap between grid points controlled by an explicit Lipschitz bound on the payoff.
Full details are given in Appendix~\ref{sec:numerical_methods}.

\subsection{Illustrative Example: Warm-up Costs}
\label{sec:warmup_example}

Theorem~\ref{thm:main} establishes quantum advantage under the assumption that cumulative baseline output $T(t)$ is strictly convex.
We now construct a physically motivated example satisfying this condition and derive the resulting throughput--splitting probability relationship.

For tasks with warm-up dynamics, such as cache population, JIT compilation, or context acquisition, instantaneous productivity $\phi(t)$ increases over uninterrupted work periods~\cite{hennessy2011computer}.
A natural model is exponential saturation:
\begin{equation}
\phi(t) = \phi_{\max}(1 - e^{-\alpha t}),
\end{equation}
corresponding to first-order relaxation toward steady-state productivity $\phi_{\max}$, with warm-up rate $\alpha$.
Integrating with $T(0) = 0$ yields
\begin{equation}
T(t) = \phi_{\max}\left[t - \frac{1}{\alpha}(1 - e^{-\alpha t})\right],
\end{equation}
which is strictly convex since $T''(t) = \phi_{\max}\alpha e^{-\alpha t} > 0$.

To obtain long-run average throughput, recall that each server alternates between idle periods $I \sim \mathrm{Exp}(\Lambda)$, during which it processes baseline work, and busy periods $B$ serving customers.
The batch arrival rate is $\Lambda = \lambda(1+p)/2$.
By the renewal reward theorem~\cite{grimmett2020probability},
\begin{equation}
\mathcal{T}(p) = \frac{\mathbb{E}[T(I)]}{\mathbb{E}[I + B]}.
\end{equation}
Evaluating the numerator using integration by parts (Appendix~\ref{sec:baseline_throughput}) and substituting the mean cycle length gives
\begin{equation}
\label{eq:throughput_model}
\mathcal{T}(p) = \phi_{\max}(1-\rho) \cdot \frac{2\alpha}{\lambda(1+p) + 2\alpha}.
\end{equation}
This expression decreases in $p$, as required: more frequent splitting interrupts baseline work before reaching high-productivity steady state.
Normalizing by the maximum achievable throughput $\phi_{\max}(1-\rho)$, corresponding to instantaneous warm-up ($\alpha \to \infty$), yields the dimensionless throughput $\mathcal{T}(p)/[\phi_{\max}(1-\rho)] \in (0,1]$.

Figure~\ref{fig:tradeoff} shows the waiting time--throughput trade-off using this model.
Quantum strategies achieve lower waiting time than classical strategies throughout the certified advantage region, with a maximum reduction in the waiting time gap (relative to the $w$-threshold policy) of approximately 21\% near $p \approx 0.20$.

\begin{figure}[t]
\centering
\includegraphics[width=0.8\columnwidth]{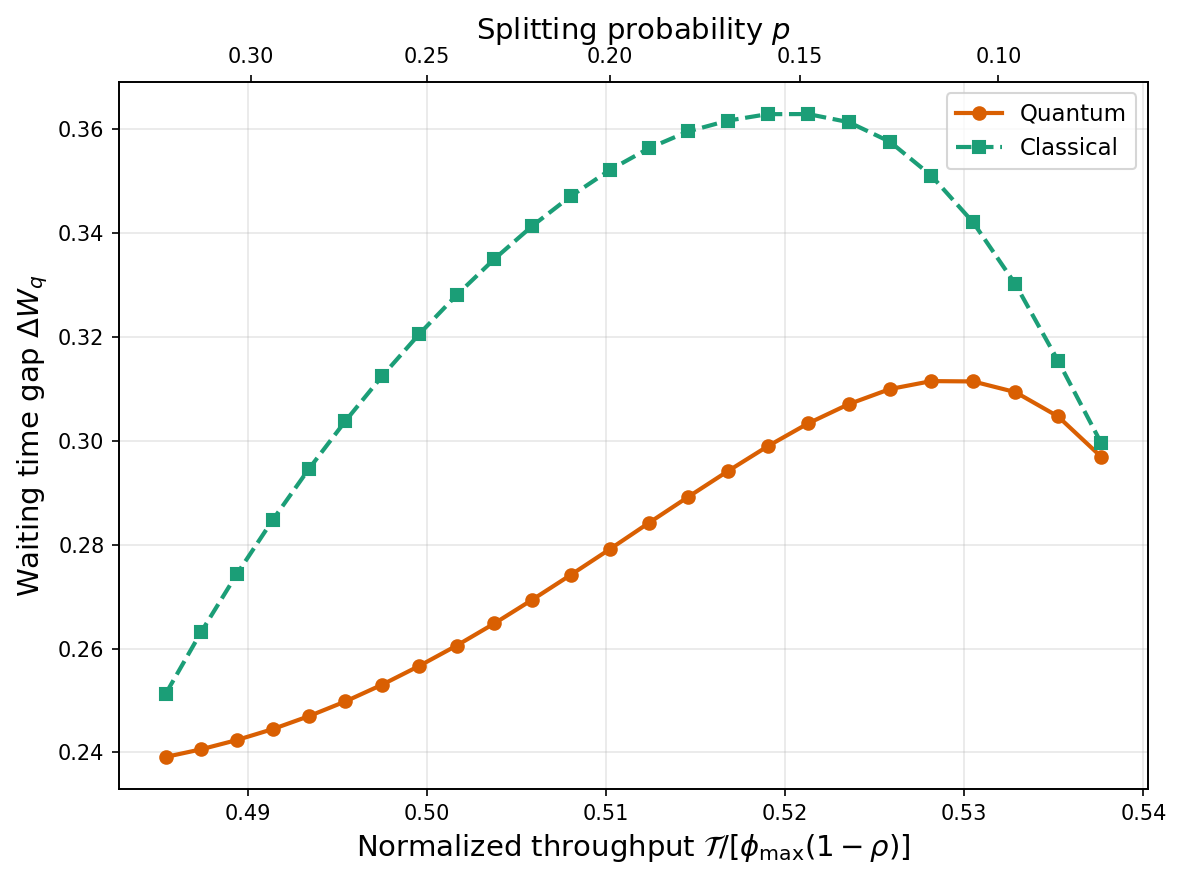}
\caption{Waiting time gap $\Delta W_q = (A^* - A)/2$ versus normalized throughput for quantum and classical strategies.
Classical values are certified upper bounds; quantum values are achieved payoffs.
Bottom axis: normalized throughput $\mathcal{T}/[\phi_{\max}(1-\rho)]$; top axis: corresponding splitting probability $p$.
System parameters: $\lambda = 0.8$, $\mu = 1$, $\alpha = 0.5$.}
\label{fig:tradeoff}
\end{figure}

\subsection{Independent arrivals and other arrival processes}
\label{sec:independent_numerical}

Our analytical results rely on two simplifying assumptions: customers arrive in pairs (one at each router simultaneously), and according to a Poisson process.
We now show numerically that the quantum advantage persists when both assumptions are relaxed.

When customers arrive independently at the two routers, there is no longer a shared arrival event to align the consumption of entangled pairs.
We therefore introduce a clock at rate $\nu$: at each tick, both routers measure their half of the next entangled pair regardless of buffer state, so pair indices stay synchronised.
If a buffer is empty the local outcome is discarded; if both have a customer at the head, the joint outcome routes them.

The clock rate $\nu$ trades coordination against latency.
Stability requires $\nu > \lambda$.
Low $\nu$ lets buffers fill, raising the probability both routers have a customer at a tick (so entanglement is consumed productively), at the cost of longer in-buffer waits.
As $\nu \to \infty$, ticks almost always find at least one buffer empty and the coordination advantage vanishes.
Together with the splitting probability $p$, this defines a two-dimensional operating space $(\nu, p)$, which we sweep numerically.
At each operating point, we simulate the quantum strategy (degree-2 polynomial Bell-measurement angles, optimised on Poisson arrivals and Exp service) and a classical-threshold-clock rule tuned to the same target $p$, using a purpose-built discrete-event simulator.
Classical strategies do not strictly require a clock, as a classical router can decide on arrival, but introducing one lets us trace the same $(\mathcal{T}, W_q)$ trade-off curve as the quantum strategy, enabling a like-for-like comparison.

Beyond independent Poisson, we consider two arrival processes chosen to span very different arrival statistics: at one extreme, fully regular arrivals; at the other, bursty traffic.
Under deterministic arrivals, customers arrive at fixed intervals of $1/\lambda$, with no randomness in the inter-arrival times.
For the bursty regime we use a two-state Markov-Modulated Poisson Process (MMPP-2)~\cite{fischer1993markov,heffes1986markov}, a standard model for correlated traffic in communication networks.
The process alternates between an active state, in which customers arrive as a Poisson process of rate $2\lambda$, and a silent state, in which no customers arrive.
Transitions between the two states happen at a symmetric rate $q$, so each state has mean duration $1/q$; smaller $q$ corresponds to longer bursts and longer silences.
The two states are equiprobable in steady-state, so the long-run mean arrival rate is $\lambda$, matching the Poisson and deterministic processes.
Note that the quantum strategies used in this sweep are the same Poisson + Exp-optimised angles applied unchanged: this is a robustness check, not a per-distribution re-optimisation.

\begin{figure*}[!ht]
  \centering
  \begin{subfigure}[t]{0.48\textwidth}
    \includegraphics[width=\linewidth]{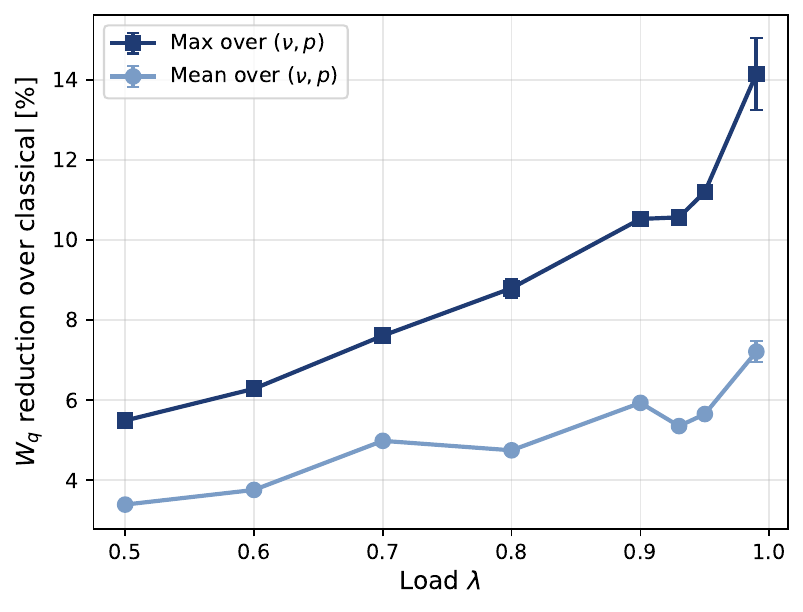}
    \caption{Max and mean waiting-time reduction across the 30 operating points $(\nu, p)$ as a function of load $\lambda$, with independent Poisson arrivals and exponential service.
    The quantum advantage broadly grows with $\lambda$, with max savings ranging from $\sim\!5\%$ at $\lambda=0.5$ to $\sim\!14\%$ at $\lambda=0.99$, and mean savings from $\sim\!3\%$ to $\sim\!7\%$.}
    \label{fig:panel_a}
  \end{subfigure}\hfill
  \begin{subfigure}[t]{0.48\textwidth}
    \includegraphics[width=\linewidth]{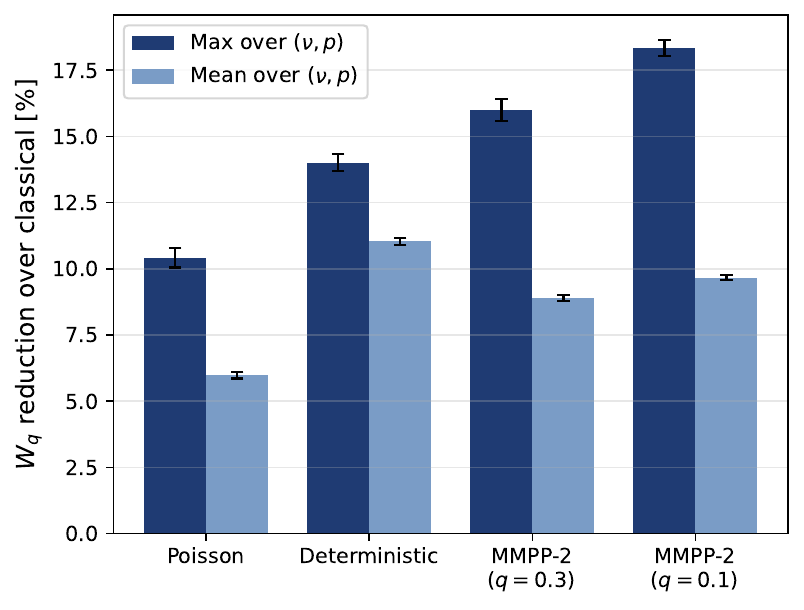}
    \caption{Max and mean waiting-time reduction across the 30 operating points $(\nu, p)$ at $\lambda = 0.9$, for four arrival processes.
    Burstier processes yield larger advantage: max ranges from $\sim\!10\%$ (Poisson) to $\sim\!18\%$ (MMPP-2 with $q=0.1$).}
    \label{fig:panel_b}
  \end{subfigure}
  \caption{Quantum advantage with independent arrivals.
  (Left) The advantage holds across loads and broadly grows with $\lambda$.
  (Right) It holds across arrival processes and grows with arrival burstiness.
  The 30 operating points are $\nu/\lambda \in \{1.2, 1.5, 2, 3, 5\}$ and $p \in \{0.10, 0.15, 0.20, 0.25, 0.30, 0.40\}$.
  Error bars are paired-bootstrap standard errors over simulation seeds.}
  \label{fig:numerical}
\end{figure*}

Figure~\ref{fig:numerical} summarises the results.
The advantage holds with independent rather than paired arrivals, across the full range of $\lambda$ tested, and broadly grows with load (Figure~\ref{fig:panel_a}).
Further, it is largest for bursty arrivals, which we expect are more representative of real network traffic (Figure~\ref{fig:panel_b}).


\section{Discussion and Outlook}
\label{sec:discussion}
Advantages in non-local games are typically modest in magnitude.
The canonical example is the CHSH game, where quantum strategies achieve a winning probability of approximately $85\%$ versus $75\%$ for classical strategies, a relative improvement of roughly $13\%$~\cite{clauser1969proposed}.
In our weighted routing game, the maximum payoff gap represents a comparable relative advantage.
Such gaps might suggest that quantum coordination offers only marginal practical benefit.

However, the relevant metric for system performance is not the probability of the players' decisions agreeing with the winning condition, but rather operational quantities, such as waiting time and throughput, that may depend non-linearly on coordination quality.
This non-linearity can amplify modest coordination improvements into substantial operational gains.
In our model, this happens through two distinct mechanisms.

First, the waiting time depends on the second moment of the service time distribution, making misrouting of high-service-time pairs disproportionately costly.
Our routing game captures this directly: unlike standard non-local games where all disagreements with the winning condition contribute equally, the payoff weights each routing decision by its operational impact $w(x_1, x_2)$.

Second, when the baseline throughput function $T(t)$ is strictly convex, throughput grows super-linearly with uninterrupted processing time.
Better coordination at fixed waiting time enables operation at lower splitting probability, with the curvature of $T(t)$ amplifying small reductions in $p$ into larger throughput gains.
The magnitude of this amplification depends on how severely the baseline task penalizes interruption.

This suggests a general principle: quantum coordination may offer meaningful practical advantage in settings where operational metrics depend super-linearly on coordination quality.
In queueing and scheduling contexts, this corresponds to tasks with significant warm-up costs, batch processing requirements, or context-switching penalties.

The quantum resources required to demonstrate this advantage are modest.
The protocol uses bipartite entanglement shared between the two decision makers and local measurements conditioned on the observed service times.
Each routing decision consumes at most one shared entangled pair, so in steady state the entanglement consumption rate is of the same order as the routing decision rate in the target system.
Heralded entanglement generation and storage have been demonstrated in several physical platforms~\cite{covey2023quantum,ruf2021quantum,krutyanskiy2023entanglement}, including over deployed metropolitan fiber~\cite{stolk2024metropolitan}.
Our analysis relies on two assumptions: (i) an entangled pair is present when a routing decision is made, and (ii) memory readout is deterministic.
Both hold for heralded entanglement with near-unit efficiency memory.
Unheralded schemes (e.g., SPDC-based) do not guarantee (i), and low-efficiency memories do not guarantee (ii).

In this work we assumed perfect, always-available entanglement.
In practice, non-unit fidelity reduces the achievable quantum payoff, while limited entanglement availability (due to generation rate or memory lifetime) causes the effective policy to interpolate between quantum and classical performance.
Quantifying the fidelity and availability thresholds required to retain meaningful advantage is an important direction for future work.

For our proofs, we assumed paired simultaneous Poisson arrivals for ease of analysis.
Our numerical results indicate that quantum advantage persists for independent arrivals, and is larger for burstier arrivals, which are more representative of network traffic.

Our analysis targets settings where routing decisions must be taken faster than classical round-trip communication allows, so our baseline explicitly rules out real-time state exchange.
Comparing against latency-constrained classical protocols that use delayed or coarse-grained state information would require separate analysis.
Promising application domains include wide-area content-delivery networks and wireless medium-access control, where uninterrupted background processing is disproportionately valuable.
\section{Conclusion}
\label{sec:conclusion}
We have shown that entanglement-assisted coordination yields rigorously certified performance gains in a concrete distributed routing model.
Specifically, we characterized the optimal full-information routing policy as a weighted threshold rule, formulated the coordination problem as a weighted non-local game with payoff directly determining waiting time, proved that optimal classical strategies take a threshold form enabling efficient certification, and demonstrated quantum advantage across a substantial parameter range.

These results suggest distributed scheduling and load balancing as a candidate application domain for near-term quantum networks, offering a starting point for translating non-local correlations into operational improvements in latency-constrained distributed systems.

\section*{Code Availability}
The code used to perform the numerical optimizations and generate the figures in this paper is available at \url{https://gitlab.com/FranciscoHS/entanglement-routing-advantage}.

\section*{Acknowledgements}
We thank Jeroen Grimbergen and John Gardiner for useful discussions on the mathematical analysis.
We thank Conor E. Bradley, Janice van Dam and Scarlett Gauthier for critical reading of earlier versions of the manuscript.
We acknowledge funding from the Dutch Research Council (NWO) through the project “QuTech Part III Application based research - Demonstrators” (project number 601.QT.001), and from the European Union under project InQubate-101213305.
Views and opinions expressed are however those of the author(s) only and do not necessarily reflect those of the European Union or the European Innovation Council.
Neither the European Union nor the granting authority can be held responsible for them.

\section*{Competing Interests}
The authors declare no competing financial or non-financial interests.

\section*{Author Contributions}
F.F.S. and S.W. conceptualized the project.
F.F.S. developed the queueing-theoretic analysis, proved the theoretical results, performed the numerical analysis, and wrote the manuscript.
All authors revised the manuscript.

%
 
\onecolumngrid
\appendix*
\onecolumngrid
\appendix
\section{The $w$-threshold policy is Pareto optimal}
\label{sec:delta_threshold}
In this appendix we present a rigorous analysis of the queueing model introduced in the main text, culminating in a proof that the $w$-threshold policy achieves Pareto optimality: no other policy can improve one objective (reduce waiting time or increase throughput) without worsening the other.

A central quantity in our analysis is the \emph{splitting benefit function}
\begin{equation}
w(x_1, x_2) = c_1 x_1 x_2 + c_2(x_1 + x_2),
\end{equation}
where $c_1 = \frac{\lambda}{2(1-\rho)}$ and $c_2 = \frac{1}{4}$.
This function quantifies the reduction in expected customer waiting time achieved by splitting a pair with service times $(x_1, x_2)$ rather than bunching them.
The term $c_1 x_1 x_2$ captures the queueing benefit from reducing workload variance, while $c_2(x_1 + x_2)$ captures the delay avoided when customers are sent to separate servers rather than queued behind one another.
The optimal routing policy, as we shall prove, takes a simple threshold form: split pairs for which $w(x_1, x_2)$ exceeds a threshold $\tau$, and bunch pairs for which it falls below.

We structure the appendix as follows.
First, we formalize the system model and establish key symmetry properties.
We then reduce the system to a single-server perspective, define the core quantities and decompose customer waiting time into interpretable components.
With this foundation, we formulate the routing problem as a constrained optimization over splitting decisions and prove that the $w$-threshold policy is optimal at any fixed splitting probability.

\subsection{Model}
\label{app:queueing_model}
We consider a distributed system consisting of two identical servers and two routers.
Each server has a queue of unlimited capacity to hold incoming requests and follows a first-come, first-served (FCFS) discipline.
The system handles two types of work.
First, a continuously available, preemptible baseline task.
Second, servicing customers.
When a customer is assigned to a server processing the baseline task, the server immediately preempts the baseline work and begins servicing the customer.
Customers arrive in pairs, one at each router, according to a Poisson process with rate $\lambda$.
Each customer's service time is drawn independently from an exponential distribution Exp($\mu$) (rate $\mu$, mean $1/\mu$).
We denote a customer pair's service times as $(X_1,X_2)$.

Servers are never idle.
When a server completes service for all customers in its queue and finds the queue empty, it immediately resumes processing the baseline task until the next customer arrival.
Note that baseline work does not affect customer waiting-time calculations.
It only fills customer-less periods.
However, the baseline throughput depends on the distribution of the lengths of these `idle' periods (e.g., due to setup costs or non-linear returns on uninterrupted baseline time), creating a trade-off with customer waiting times.

Upon arrival of a customer pair, each router decides to which server to route its incoming customer.
We assume each router observes the service time of its customer, but has no other information.
They do not know the service time of the other customer in the pair, nor do they have information about the state of the server and its queue.
Under these constraints, an implementable routing policy is a map from a single observed service time to a choice of server.
To provide a fundamental bound on performance, we analyze \emph{oracle} routing policies that have access to both service times $(X_1, X_2)$ when making routing decisions.
These policies provide a lower bound on waiting time achievable by any implementable policy, which observes only a single service time.
An oracle policy maps each pair $(x_1, x_2)$ to a distribution over four routing actions: split with $\min(x_1,x_2)$ to Server~1, split with $\max(x_1,x_2)$ to Server~1, bunch to Server~1, or bunch to Server~2.
We parameterize policies by the splitting probability $r(x_1, x_2) \in [0,1]$, the probability that a pair with service times $(x_1, x_2)$ is split.
The overall splitting probability is $p := \mathbb{E}[r(X_1,X_2)]$.
When writing expectations involving the routing policy, we use the shorthand $\mathbb{E}[r \cdot f(X_1, X_2)]$ to denote $\mathbb{E}[r(X_1, X_2) \cdot f(X_1, X_2)]$.

When two customers are sent to the same server, one customer is selected uniformly at random to precede the other.
Both are then processed FCFS.

We restrict attention to \emph{load-balanced} policies, which treat the two servers identically: when splitting, each server is equally likely to receive the larger job; when bunching, the destination server is chosen uniformly at random.
This restriction is without loss of generality.

\begin{proposition}[Load Balancing]\label{prop:symmetrization}
For any oracle routing policy, there exists a load-balanced policy achieving weakly lower waiting time at the same baseline throughput.
\end{proposition}

We prove this in Appendix~\ref{sec:symmetrization}.
The key insight is that baseline throughput depends only on the splitting probability $p$, not on the server assignment when splitting or bunching (Proposition~\ref{prop:baseline-monotonicity}), while load imbalance strictly increases waiting time.
Thus load-balanced policies weakly dominate all other policies in the joint (waiting time, throughput) space, and restricting attention to them is without loss of generality for characterizing the Pareto frontier.

Under load-balanced policies, both servers see statistically identical arrival processes: each receives customers at rate $\lambda$, with identical batch size distributions and identical conditional service time distributions.
We may therefore analyze the system from the perspective of a single server.
This yields per-server utilization $\rho = \lambda/\mu$, independent of the splitting probability.
We assume $\rho < 1$ throughout, ensuring stability.

From this single-server perspective, customers arrive according to a compound Poisson process with batch arrival rate
\begin{equation}
\Lambda = \lambda\frac{1+p}{2}
\end{equation}
and batch size distribution (for $K\in\{1,2\}$)
\begin{equation}
\Pr\{K=1\} = \frac{2p}{1+p}, \qquad \Pr\{K=2\} = \frac{1-p}{1+p}.
\end{equation}
These expressions follow from the splitting/bunching structure.
The batch arrival rate $\Lambda = \lambda(1+p)/2$ is obtained by averaging $\lambda(1+r(x_1,x_2))/2$ (the rate at which batches are generated at this server) over $(X_1, X_2)$.
The batch size probabilities are then obtained by normalizing the rates of size-1 arrivals (from split pairs, rate $\lambda p$) and size-2 arrivals (from bunched pairs sent here, rate $\lambda(1-p)/2$).
Note that the customer arrival rate at a single server is $\Lambda \mathbb{E}[K] = \lambda$, as required by symmetry.
Because the routing policy depends on $(X_1,X_2)$, the per-customer service time distribution at a server is generally non-exponential; hence the single-server model is $M^X/G/1$ (compound Poisson arrivals with general service times).
Note that for batches of size $K=2$, the service times $(X_1, X_2)$ are generally not independent, as their presence in the same batch is conditioned on the routing policy $r(x_1, x_2)$.

\subsection{Single-server primitives and batch workload}
\label{app:single_server}
We now derive the second moment of the batch workload seen by a single server, which is the key input to the Pollaczek--Khinchine formula~\cite{gross2011fundamentals} for computing expected waiting time.

We apply the Pollaczek--Khinchine (PK) formula for an $M^X/G/1$ queue, which expresses the expected virtual waiting time, i.e., the time until a batch begins service, as:
\begin{equation}
\mathbb{E}[W_q^{\text{virtual}}] = \frac{\Lambda \mathbb{E}[B_s^2]}{2(1-\rho)},
\end{equation}
where $\Lambda$ is the batch arrival rate and $\mathbb{E}[B_s^2]$ is the second moment of the batch workload.
The full customer waiting time additionally includes within-batch delay, i.e., the time one customer spends waiting for its partner in a batch, which we derive in Appendix~\ref{sec:decomposition}.
We have already established that $\Lambda = \lambda (1+p)/2$.
In this section, we derive the product $\Lambda \mathbb{E}[B_s^2]$.

For a customer pair with service times $(X_1, X_2)$, define the total service time $S = X_1 + X_2.$
We require only that $S$ has finite first and second moments, which is satisfied under individual customer service times Exp($\mu$).

Rather than computing $\mathbb{E}[B_s^2]$ directly (which would require conditioning on a batch arriving at this server), we compute the expected squared-workload contribution from each arriving pair.
For a pair $(x_1, x_2)$ with splitting probability $r = r(x_1, x_2)$, define $Y$ to be the squared workload contributed to this server. The possible outcomes are:
\begin{itemize}
    \item With probability $(1-r)/2$, the pair is bunched to this server, contributing $(x_1 + x_2)^2$.
    \item With probability $r/2$, the pair is split and $x_1$ is routed here, contributing $x_1^2$.
    \item With probability $r/2$, the pair is split and $x_2$ is routed here, contributing $x_2^2$.
    \item With probability $(1-r)/2$, the pair is bunched to the other server, contributing $0$.
\end{itemize}
Thus
\begin{equation}
    \mathbb{E}[Y | x_1, x_2] = \frac{1-r}{2}(x_1+x_2)^2 + \frac{r}{2}x_1^2 + \frac{r}{2}x_2^2 = \frac{1}{2}\left[(1-r)(x_1+x_2)^2 + r(x_1^2+x_2^2)\right].
\end{equation}
Since pairs arrive at rate $\lambda$, the expected sum of squared batch workloads per unit time is $\lambda \mathbb{E}[Y]$, which equals $\Lambda \mathbb{E}[B_s^2]$. Therefore
\begin{equation}
    \Lambda\,\mathbb{E}[B_s^2] = \frac{\lambda}{2}\,\mathbb{E}\!\left[(1-r)(X_1+X_2)^2+r\,(X_1^2+X_2^2)\right].
\end{equation}
Rearranging,
\begin{align}
    \Lambda\,\mathbb{E}[B_s^2]
    &=\frac{\lambda}{2}\,\mathbb{E}\!\left[(X_1+X_2)^2 - r\left((X_1+X_2)^2 - X_1^2-X_2^2\right)\right]\\
    &=\frac{\lambda}{2}\Big(\mathbb{E}[S^2]-2\mathbb{E}[r\,X_1X_2]\Big).
\end{align}

\subsection{Waiting time decomposition}
\label{sec:decomposition}
We now characterize the mean customer waiting time under a given routing policy.
From the single-server perspective established in Appendix~\ref{app:queueing_model}, we have an $M^X/G/1$ queue with batch arrival rate $\Lambda = \lambda(1+p)/2$ and utilization $\rho = \lambda/\mu$.

The customer waiting time can be decomposed into two components.
First, the virtual wait: the time a customer waits for work already in the system when its batch arrives.
Second, the within-batch delay: the additional time a customer waits if it is the second customer in a size-2 batch, in which case it must wait for its partner's service to complete.

The virtual wait component is given by the Pollaczek--Khinchin formula for $M^X/G/1$ queues:
\begin{equation}
\mathbb{E}[W_q^{\text{virtual}}] = \frac{\Lambda \mathbb{E}[B_s^2]}{2(1-\rho)}.
\end{equation}
From Appendix~\ref{app:single_server}, we have $\Lambda \mathbb{E}[B_s^2] = \frac{\lambda}{2}(\mathbb{E}[S^2] - 2\mathbb{E}[rX_1X_2])$, so
\begin{equation}
\mathbb{E}[W_q^{\text{virtual}}] = \frac{\lambda(\mathbb{E}[S^2] - 2\mathbb{E}[rX_1X_2])}{4(1-\rho)}.
\end{equation}
Note that $\rho = \lambda/\mu$ does not depend on the policy.

For the within-batch delay, consider a given pair $(x_1, x_2)$.
If the pair is bunched (probability $1-r$), one of the two customers waits for the other's service.
The expected within-batch delay for a randomly selected customer from this pair is $(x_1 + x_2)/4$: with probability $1/2$ the customer is second and waits for an expected time of $(x_1 + x_2)/2$.
Therefore, the expected within-batch delay contribution from a pair $(x_1, x_2)$ is $(1-r)(x_1+x_2)/4$.
Averaging over all pairs yields
\begin{equation}
\mathbb{E}[\text{within-batch delay}]
= \tfrac{1}{4}\,\mathbb{E}[(1-r)S],
\end{equation}
where $S = X_1 + X_2$.

The total mean waiting time is therefore
\begin{equation}
\mathbb{E}[W_q] = \mathbb{E}[W_q^{\text{virtual}}] + \mathbb{E}[\text{within-batch delay}] = \frac{\lambda(\mathbb{E}[S^2] - 2\mathbb{E}[rX_1X_2])}{4(1-\rho)} + \frac{1}{4}\mathbb{E}[(1-r)S].
\end{equation}

\subsection{Optimization formulation at fixed splitting probability}
\label{sec:form_fixed_prob}
We now formulate the routing problem as a constrained optimization.
Recall from Appendix~\ref{sec:decomposition} that
\begin{equation}
\mathbb{E}[W_q] = \frac{\lambda(\mathbb{E}[S^2] - 2\mathbb{E}[rX_1X_2])}{4(1-\rho)} + \frac{1}{4}\mathbb{E}[(1-r)S].
\end{equation}
Expanding the within-batch term as $\mathbb{E}[(1-r)S] = \mathbb{E}[S] - \mathbb{E}[rS]$, we obtain
\begin{equation}
\mathbb{E}[W_q] = \frac{\lambda\mathbb{E}[S^2]}{4(1-\rho)} + \frac{\mathbb{E}[S]}{4} - \frac{\lambda}{2(1-\rho)}\mathbb{E}[rX_1X_2] - \frac{1}{4}\mathbb{E}[rS].
\end{equation}

Define constants
\begin{equation}
c_1 = \frac{\lambda}{2(1-\rho)}, \qquad c_2 = \frac{1}{4}.
\end{equation}
Then we can write
\begin{equation}
\mathbb{E}[W_q] = \text{const} - \big(c_1\,\mathbb{E}[r X_1 X_2] + c_2\,\mathbb{E}[rS]\big),
\end{equation}
where $\text{const} = \frac{\lambda\mathbb{E}[S^2]}{4(1-\rho)} + \frac{\mathbb{E}[S]}{4}$ is independent of the routing policy.

Since minimizing $\mathbb{E}[W_q]$ is equivalent to maximizing the term in parentheses, we can formulate the routing problem at fixed splitting probability $p \in [0,1]$ as:

\begin{problem}[Optimal routing at fixed splitting probability; restates Problem~\ref{prob:fixed_p_main}]\label{prob:fixed_p}
Given a target splitting probability $p \in [0,1]$, find a routing policy $r: \mathbb{R}_+^2 \to [0,1]$ that solves
\begin{equation}
\begin{aligned}
\text{maximize} \quad &c_1\,\mathbb{E}[r X_1 X_2] + c_2\,\mathbb{E}[rS] \\
\text{subject to} \quad & \mathbb{E}[r] = p, \\
& 0 \leq r(x_1,x_2) \leq 1 \text{ for all } (x_1,x_2).
\end{aligned}
\end{equation}
\end{problem}

 \subsection{Optimal policy}
\label{sec:optimal_policy}

We now characterize the solution to Problem~\ref{prob:fixed_p}.
Recall the splitting benefit function $w: \mathbb{R}_+^2 \to \mathbb{R}_+$ defined at the beginning of this appendix:
\begin{equation}
w(x_1, x_2) = c_1 x_1 x_2 + c_2(x_1 + x_2),
\end{equation}
where $c_1 = \frac{\lambda}{2(1-\rho)}$ and $c_2 = \frac{1}{4}$. 
The objective function can thus be written compactly as $\mathbb{E}[r \cdot w(X_1, X_2)]$.
Since $X_1, X_2 \sim \text{Exp}(\mu)$, all moments exist, and $w(X_1, X_2)$ is integrable, i.e., $w \in L^1$.

\begin{theorem}[Optimal threshold policy]\label{thm:optimal_threshold}
For any target splitting probability $p \in [0,1]$, the optimal routing policy for Problem~\ref{prob:fixed_p} is a threshold policy:
\begin{equation}
r^*(x_1,x_2) = \begin{cases}
1 & \text{if } w(x_1,x_2) > \tau_p, \\
0 & \text{if } w(x_1,x_2) < \tau_p,
\end{cases}
\end{equation}
where the threshold $\tau_p$ is chosen such that $\mathbb{E}[r^*] = p$. Under the exponential service time distribution, the random variable $W = w(X_1, X_2)$ is continuously distributed; thus $\mathbb{P}(w = \tau_p) = 0$, and the threshold $\tau_p$ uniquely determines $p$.
\end{theorem}

\begin{proof}
Let $r^*$ be the threshold policy defined above, and let $r$ be any other feasible policy such that $\mathbb{E}[r] = p$.
We wish to show $\mathbb{E}[r^* w] \geq \mathbb{E}[r w]$.

For any realization $(x_1, x_2)$, consider the expression $(r^*(x_1, x_2) - r(x_1, x_2))(w(x_1, x_2) - \tau_p)$. We observe the following:
\begin{itemize}
    \item If $w(x_1, x_2) > \tau_p$, then $r^* = 1$. Since $r \in [0, 1]$, $(1 - r) \geq 0$ and $(w - \tau_p) > 0$. Their product is $\geq 0$.
    \item If $w(x_1, x_2) < \tau_p$, then $r^* = 0$. Since $r \in [0, 1]$, $(0 - r) \leq 0$ and $(w - \tau_p) < 0$. Their product is $\geq 0$.
    \item If $w(x_1, x_2) = \tau_p$, the product is $0$.
\end{itemize}
Thus, $(r^* - r)(w - \tau_p) \geq 0$ holds pointwise for all $(x_1, x_2)$. Taking expectations:
\begin{equation}
\mathbb{E}[(r^* - r)(w - \tau_p)] \geq 0.
\end{equation}
Expanding the expectation:
\begin{equation}
\mathbb{E}[r^* w] - \mathbb{E}[r w] - \tau_p (\mathbb{E}[r^*] - \mathbb{E}[r]) \geq 0.
\end{equation}
By the constraint $\mathbb{E}[r^*] = \mathbb{E}[r] = p$, the last term is zero. Therefore, $\mathbb{E}[r^* w] \geq \mathbb{E}[r w]$, proving that the threshold policy $r^*$ is optimal.
\end{proof}

The threshold structure arises because the constraint $\mathbb{E}[r] = p$ acts like a budget: we may split at most a fraction $p$ of pairs on average.
Optimality requires allocating this budget to pairs with highest marginal benefit $w$.

As the target splitting probability $p$ increases, the threshold $\tau_p$ decreases, allowing more pairs to be split.
This family of policies, indexed by $p$, will be shown to trace the Pareto frontier between customer waiting time and the baseline throughput (which is a function of $p$).

\subsection{Baseline task throughput}
\label{sec:baseline_throughput}
Having characterized the policy minimizing waiting time at fixed splitting probability, we now analyze the second objective---baseline throughput---to establish the Pareto trade-off.

Following the single-server perspective established in Appendix~\ref{app:queueing_model}, we analyze one server's baseline throughput.
Recall that each server processes a continuously available baseline task when it has no customers to service, switching between customer and baseline work instantly with no overhead.

The server alternates between idle periods, when no customers are present, and busy periods, when at least one is.
Let $I$ denote the duration of an idle period and $B$ denote the duration of a busy period.

During an idle period, the server processes baseline work until the next batch arrival.
Since batches arrive according to a Poisson process with rate $\Lambda = \lambda(1+p)/2$, the duration of an idle period is exponentially distributed:
\begin{equation}
I \sim \text{Exp}(\Lambda), \qquad \mathbb{E}[I] = \frac{1}{\Lambda} = \frac{2}{\lambda(1+p)}.
\end{equation}

For the busy period, we derive the mean duration using a rate-balance argument.
By PASTA (Poisson Arrivals See Time Averages)~\cite{wolff1982poisson}, a batch arriving according to a Poisson process finds the system in steady state: idle with probability $1-\rho$ and busy with probability $\rho$.
Busy periods are therefore initiated at rate $\Lambda(1-\rho)$---the batch arrival rate times the probability of finding an idle server.
Since the long-run fraction of time spent busy equals the rate of entering the busy state times the mean time spent there, we have $\rho = \Lambda(1-\rho) \cdot \mathbb{E}[B]$, giving $\mathbb{E}[B] = \rho / [\Lambda(1-\rho)]$.
Substituting $\rho=\lambda/\mu$ and $\Lambda=\lambda(1+p)/2$,
\begin{equation}
\mathbb{E}[B] = \frac{\rho}{\Lambda(1-\rho)} = \frac{2}{(1+p)(\mu-\lambda)},
\end{equation}
which decreases with $p$, implying that higher splitting probability leads to shorter busy periods.

The system forms a renewal process: since the system returns to the same empty state at the end of each busy period, and Poisson arrivals are memoryless, successive cycles are independent.
The mean cycle length is
\begin{equation}
\mathbb{E}[I+B] = \frac{2}{\lambda(1+p)} + \frac{2}{(1+p)(\mu-\lambda)} = \frac{2}{1+p}\left(\frac{1}{\lambda} + \frac{1}{\mu-\lambda}\right).
\end{equation}

Let $T(t)$ denote the amount of baseline work completed during an idle period of length $t$.
We assume $T$ is increasing (longer idle periods complete more work), differentiable with $T(0) = 0$, and denote its derivative by $\phi = T'$.

To compute the long-run average baseline throughput, we apply the renewal reward theorem~\cite{grimmett2020probability}:
\emph{the long-run average reward per unit time equals the expected reward per cycle divided by the expected cycle length.}
Here, the ``reward'' in each cycle is the baseline work $T(I)$ completed during the idle period.
Therefore, the long-run average baseline throughput per server is:
\begin{equation}
\mathcal{T}(p) = \frac{\mathbb{E}[T(I)]}{\mathbb{E}[I + B]},
\end{equation}
where
\begin{equation}
\mathbb{E}[T(I)] = \int_0^\infty T(t) \Lambda e^{-\Lambda t}\,dt.
\end{equation}

Note that baseline throughput depends only on the splitting probability $p$, not on which specific pairs are split.
\begin{proposition}[Monotonicity of baseline throughput in $p$]\label{prop:baseline-monotonicity}
Let $\Lambda=\lambda(1+p)/2$ and assume $T$ is differentiable with derivative $\phi=T'$ and $T(0) = 0$, with $T$ increasing. Then:
\begin{enumerate}
\item If $T$ is strictly convex (so $\phi$ is increasing), then $\mathcal{T}(p)$ strictly decreases in $p$.
\item If $T$ is linear (so $\phi$ is constant), then $\mathcal{T}(p)$ is independent of $p$.
\item If $T$ is strictly concave (so $\phi$ is decreasing), then $\mathcal{T}(p)$ strictly increases in $p$.
\end{enumerate}
\end{proposition}

\begin{proof}
Write
\[
\mathcal{T}(p)=\frac{\mathbb{E}[T(I)]}{\mathbb{E}[I+B]}
=\frac{1+p}{C}\,\mathbb{E}[T(I)],
\qquad
C=2\left(\frac{1}{\lambda}+\frac{1}{\mu-\lambda}\right).
\]
Set
$f(\Lambda):=\mathbb{E}[T(I)]$ so that
\[
\mathcal{T}(p)=\frac{1+p}{C}\,f(\Lambda).
\]
Differentiate with respect to $p$:
\[
\frac{d\mathcal{T}}{dp}
=\frac{1}{C}\Big(f(\Lambda)+(1+p)\,f'(\Lambda)\,\frac{d\Lambda}{dp}\Big).
\]
Since $\frac{d\Lambda}{dp}=\frac{\lambda}{2}$ and $\Lambda=\frac{\lambda}{2}(1+p)$, we have
$(1+p)\,\frac{d\Lambda}{dp}=\Lambda$. Hence
\[
\frac{d\mathcal{T}}{dp}
=\frac{1}{C}\big(f(\Lambda)+\Lambda f'(\Lambda)\big)
=\frac{1}{C}\,\frac{d}{d\Lambda}\big(\Lambda f(\Lambda)\big)
=\frac{1}{C}\,\frac{d}{d\Lambda}\big(\Lambda\,\mathbb{E}[T(I)]\big).
\]
Because $C>0$, the sign of $\frac{d\mathcal{T}}{dp}$ equals the sign of
$\frac{d}{d\Lambda}\big(\Lambda\,\mathbb{E}[T(I)]\big)$.
With $I\sim\mathrm{Exp}(\Lambda)$, we use integration by parts. Since $T(0) = 0$ and $T$ is differentiable,
\[
\mathbb{E}[T(I)]=\int_{0}^{\infty} T(t)\,\Lambda e^{-\Lambda t}\,dt
= \left[T(t)(-e^{-\Lambda t})\right]_0^\infty + \int_0^\infty \phi(t) e^{-\Lambda t}\,dt
=\int_{0}^{\infty} \phi(t)\,e^{-\Lambda t}\,dt,
\]
hence
\[
g(\Lambda):=\Lambda\,\mathbb{E}[T(I)]
=\Lambda\int_{0}^{\infty}\phi(t)e^{-\Lambda t}\,dt,
\quad
g'(\Lambda)=\int_{0}^{\infty}\phi(t)e^{-\Lambda t}(1-\Lambda t)\,dt.
\]
Let $U\sim\mathrm{Exp}(\Lambda)$, so $\mathbb{E}[U] = 1/\Lambda$. 
Rewriting the integrals in terms of expectations with respect to $U$:
\[
\int_{0}^{\infty}\phi(t)e^{-\Lambda t}\,dt 
= \frac{1}{\Lambda}\int_{0}^{\infty}\phi(t)\Lambda e^{-\Lambda t}\,dt 
= \frac{1}{\Lambda}\,\mathbb{E}[\phi(U)],
\]
and
\[
\int_{0}^{\infty}\phi(t)\,t\,\Lambda e^{-\Lambda t}\,dt = \mathbb{E}[U\,\phi(U)].
\]
Therefore,
\[
g'(\Lambda)
= \frac{1}{\Lambda}\,\mathbb{E}[\phi(U)] - \mathbb{E}[U\,\phi(U)]
= \mathbb{E}[U]\,\mathbb{E}[\phi(U)] - \mathbb{E}[U\,\phi(U)]
= -\mathrm{Cov}(U,\phi(U)).
\]
Since $\Lambda > 0$, the sign of $d\mathcal{T}/dp$ equals the sign of $-\mathrm{Cov}(U, \phi(U))$.
When $\phi$ is monotonic, $U$ and $\phi(U)$ are co-monotonic, so the covariance has the same sign as $\phi'$.
Thus: (i) $\phi$ increasing (convex $T$) implies $\mathrm{Cov}(U,\phi(U)) > 0$, hence $\mathcal{T}$ is strictly decreasing in $p$; (ii) $\phi$ constant implies $\mathrm{Cov}(U,\phi(U)) = 0$, hence $\mathcal{T}$ is independent of $p$; (iii) $\phi$ decreasing implies $\mathrm{Cov}(U,\phi(U)) < 0$, hence $\mathcal{T}$ is strictly increasing in $p$.
\end{proof}

For tasks with setup costs, learning curves, or context-switching penalties, the function $T$ is convex and baseline throughput strictly decreases with splitting probability.
In what follows, we focus on the strictly convex case, which yields a nontrivial Pareto trade-off between the two objectives.

\subsection{Pareto optimality of the threshold policy}
\label{sec:pareto_optimality}

We now establish that the $w$-threshold policy characterized in Theorem~\ref{thm:optimal_threshold} achieves the Pareto frontier between customer waiting time and baseline task throughput.

Recall that any routing policy induces a splitting probability $p = \mathbb{E}[r] \in [0,1]$.
From Proposition~\ref{prop:baseline-monotonicity}, the baseline throughput $\mathcal{T}(p)$ depends only on this splitting probability, not on which specific pairs are split.
Thus characterizing the Pareto frontier reduces to minimizing customer waiting time at each fixed splitting probability.

For a fixed splitting probability $p \in [0,1]$, Theorem~\ref{thm:optimal_threshold} establishes that the threshold policy $r^*_p$ minimizes customer waiting time among all policies achieving splitting probability $p$.
We now show that as $p$ varies over $[0,1]$, these optimal policies trace out the complete Pareto frontier.

\begin{theorem}[Pareto optimality; restates Theorem~\ref{thm:pareto}]\label{thm:pareto_optimality}
Assume $T$ is strictly convex (so baseline throughput strictly decreases in $p$ by Proposition~\ref{prop:baseline-monotonicity}).
Then the family of threshold policies $\{r^*_p : p \in [0,1]\}$, where $r^*_p$ solves Problem~\ref{prob:fixed_p} for splitting probability $p$, achieves the Pareto frontier between baseline throughput and customer waiting time.
\end{theorem}

\begin{proof}
Consider any feasible routing policy $r$ with induced splitting probability $p_r = \mathbb{E}[r]$.
We show that $r$ is weakly dominated by the threshold policy $r^*_{p_r}$ achieving the same splitting probability.

By Proposition~\ref{prop:baseline-monotonicity}, baseline throughput depends only on the splitting probability:
\begin{equation}
\mathcal{T}(r) = \mathcal{T}(p_r) = \mathcal{T}(r^*_{p_r}).
\end{equation}

By Theorem~\ref{thm:optimal_threshold}, among all policies with splitting probability $p_r$, the threshold policy $r^*_{p_r}$ minimizes customer waiting time:
\begin{equation}
\mathbb{E}[W_q](r^*_{p_r}) \leq \mathbb{E}[W_q](r).
\end{equation}

Thus $r^*_{p_r}$ achieves the same baseline throughput as $r$ with no greater customer waiting time, so $r$ is weakly dominated by $r^*_{p_r}$.
This establishes that every point on the Pareto frontier corresponds to a threshold policy.

Conversely, every threshold policy $r^*_p$ is Pareto optimal.
Any policy with the same splitting probability $p$ has the same throughput (Proposition~\ref{prop:baseline-monotonicity}) and weakly higher waiting time (Theorem~\ref{thm:optimal_threshold}), so cannot strictly dominate $r^*_p$.
Any policy with different splitting probability has different throughput, so cannot improve both objectives simultaneously.

It remains to show that increasing $p$ strictly reduces waiting time under the threshold policy.
Recall from Appendix~\ref{sec:form_fixed_prob} that $\mathbb{E}[W_q] = \text{const} - \mathbb{E}[r \cdot w]$, where the constant is independent of the routing policy.
For the threshold policy $r^*_p$, we have $\mathbb{E}[r^*_p \cdot w] = \mathbb{E}[w \cdot \mathbf{1}\{w > \tau_p\}]$.
As $p$ increases, the threshold $\tau_p$ decreases (since $p = \Pr[w > \tau_p]$ and the distribution of $w$ is continuous).
Since $w \geq 0$, decreasing $\tau_p$ strictly increases $\mathbb{E}[w \cdot \mathbf{1}\{w > \tau_p\}]$, and hence strictly decreases $\mathbb{E}[W_q](r^*_p)$.

Since $T$ is strictly convex, by Proposition~\ref{prop:baseline-monotonicity}, $\mathcal{T}$ strictly decreases in $p$.
Therefore increasing $p$ strictly reduces waiting time while strictly decreasing throughput, and decreasing $p$ strictly increases waiting time while strictly increasing throughput.
Thus no feasible policy strictly dominates $r^*_p$.
\end{proof}

\begin{remark}
The endpoints $p=1$ and $p=0$ are attained by taking $\tau = 0$ and $\tau = \infty$, respectively.
\end{remark}

Theorem~\ref{thm:pareto_optimality} establishes that the $w$-threshold policy achieves the complete Pareto frontier.
The operating point can be selected by choosing an appropriate splitting probability $p$ (or, equivalently, threshold $\tau$): higher splitting probability reduces customer waiting time at the expense of baseline throughput.

\subsection{Optimality of Load-Balanced Server Assignment}
\label{sec:symmetrization}
We now prove Proposition~\ref{prop:symmetrization}, justifying the restriction to load-balanced policies.
For any routing policy that induces a splitting probability $p$, we show that load-balanced server assignment minimizes the average customer waiting time.

\begin{proposition}[Load-Balanced Dominance]
For any oracle routing policy $\pi$ with splitting probability $p$, there exists a load-balanced policy $\pi^{\text{bal}}$ with the same splitting probability that achieves weakly lower average customer waiting time and identical baseline throughput. A policy is load-balanced if, when splitting, each server is equally likely to receive the larger job, and when bunching, the destination server is chosen uniformly at random.
\end{proposition}

\begin{proof}
Consider an arriving pair $(X_1, X_2)$ with $X_i \stackrel{\text{iid}}{\sim} \text{Exp}(\mu)$.
For each pair, an oracle policy $\pi$ selects from four possible routing actions:
\begin{itemize}
\item $A_{\min}$ (Split): Send $\min(X_1, X_2)$ to Server 1 and $\max(X_1, X_2)$ to Server 2.
\item $A_{\max}$ (Split): Send $\max(X_1, X_2)$ to Server 1 and $\min(X_1, X_2)$ to Server 2.
\item $B_1$ (Bunch): Send $X_1 + X_2$ to Server 1 and $0$ to Server 2.
\item $B_2$ (Bunch): Send $0$ to Server 1 and $X_1 + X_2$ to Server 2.
\end{itemize}

Let $(\alpha, \beta, \gamma, \delta)$ be the probabilities that the oracle selects actions $A_{\min}, A_{\max}, B_1, B_2$ respectively.
The induced splitting probability is $p = \alpha + \beta$.

\begin{lemma}[Moments of exponential order statistics]
\label{lem:moments}
For $X_1, X_2 \stackrel{\text{iid}}{\sim} \text{Exp}(\mu)$:
\begin{equation}
\mathbb{E}[\min] = \frac{1}{2\mu}, \quad \mathbb{E}[\min{}^2] = \frac{1}{2\mu^2}, \quad \mathbb{E}[\max] = \frac{3}{2\mu}, \quad \mathbb{E}[\max{}^2] = \frac{7}{2\mu^2}.
\end{equation}
\end{lemma}
\begin{proof}
Since $\min(X_1, X_2) \sim \text{Exp}(2\mu)$, we have $\mathbb{E}[\min] = 1/(2\mu)$ and $\mathbb{E}[\min^2] = 2/(2\mu)^2 = 1/(2\mu^2)$.
For the maximum, the density is $f_{\max}(t) = 2\mu(e^{-\mu t} - e^{-2\mu t})$.
Using $\int_0^\infty t^n e^{-at}\, dt = n!/a^{n+1}$, direct calculation yields $\mathbb{E}[\max] = 3/(2\mu)$ and $\mathbb{E}[\max^2] = 7/(2\mu^2)$.
\end{proof}

Let $Z_1$ denote the workload contributed to Server 1 by a single arriving pair.
Using Lemma~\ref{lem:moments} and the sum moments ($\mathbb{E}[X_1+X_2] = 2/\mu$, $\mathbb{E}[(X_1+X_2)^2] = 6/\mu^2$), the load $\rho_1$ and second-moment flux $K_1 = \lambda \mathbb{E}[Z_1^2]$ at Server 1 are:
\begin{align}
\rho_1 &= \lambda (\alpha \cdot \tfrac{1}{2\mu} + \beta \cdot \tfrac{3}{2\mu} + \gamma \cdot \tfrac{2}{\mu}), \label{eq:rho1}\\
K_1 &= \lambda (\alpha \cdot \tfrac{1}{2\mu^2} + \beta \cdot \tfrac{7}{2\mu^2} + \gamma \cdot \tfrac{6}{\mu^2}). \label{eq:K1}
\end{align}

\textbf{Step 1: Work conservation.}
Every arriving pair contributes total expected work $\mathbb{E}[X_1 + X_2] = 2/\mu$ to the system, regardless of the routing action.
Thus $\rho_1 + \rho_2 = 2\lambda/\mu = 2\bar{\rho}$ for all policies, where $\bar{\rho} = \lambda/\mu$.

\textbf{Step 2: Load-balanced policies exist and achieve $\rho_1 = \rho_2$.}
Define a policy to be \emph{load-balanced} if $\alpha = \beta$ and $\gamma = \delta$, i.e., when splitting, each server is equally likely to receive the larger job, and when bunching, each server is equally likely to receive the pair.
For a target splitting probability $p$, the load-balanced policy sets $\alpha = \beta = p/2$ and $\gamma = \delta = (1-p)/2$.
Substituting into \eqref{eq:rho1}:
\begin{equation}
\rho_1 = \lambda \left( \frac{p}{2} \cdot \frac{1}{2\mu} + \frac{p}{2} \cdot \frac{3}{2\mu} + \frac{1-p}{2} \cdot \frac{2}{\mu} \right) = \frac{\lambda}{2\mu}(p + (1-p) \cdot 2) = \frac{\lambda}{\mu} = \bar{\rho}.
\end{equation}
By an analogous calculation, $\rho_2 = \bar{\rho}$.
Thus load-balanced policies achieve $\rho_1 = \rho_2 = \bar{\rho}$ for any splitting probability $p \in [0,1]$.

In contrast, policies violating $\alpha = \beta$ or $\gamma = \delta$ can induce load imbalance.
For example, with $p = 1$, the policy $\alpha = 1, \beta = 0$ (always send the smaller job to Server 1) yields $\rho_1 = \lambda/(2\mu) = \bar{\rho}/2 \neq \bar{\rho}$.

\textbf{Step 3: Load imbalance increases waiting time.}
Since the map $(\alpha, \beta, \gamma, \delta) \mapsto (\rho_1, K_1)$ is linear and the domain is a simplex, the set of achievable pairs $(\rho_1, K_1)$ is a convex polytope.
Define the efficient frontier $K^*(\rho)$ as the minimum $K_1$ achievable for a given load $\rho_1 = \rho$.

\begin{lemma}[Convexity of efficient frontier]
\label{lem:frontier}
The efficient frontier $K^*(\rho)$ is piecewise linear and convex, with $K^{*\prime}(\rho) \geq 1/\mu > 0$.
\end{lemma}
\begin{proof}
The lower boundary of the polytope is traced through vertices $V_0$ (load $0$), $V_1$ (load $\lambda/2\mu$), $V_2$ (load $3\lambda/2\mu$), and $V_3$ (load $2\lambda/\mu$), with slopes:
\begin{equation}
\frac{dK^*}{d\rho}\bigg|_{V_0 \to V_1} = \frac{1}{\mu}, \quad \frac{dK^*}{d\rho}\bigg|_{V_1 \to V_2} = \frac{3}{\mu}, \quad \frac{dK^*}{d\rho}\bigg|_{V_2 \to V_3} = \frac{5}{\mu}.
\end{equation}
Since the slopes are strictly increasing and bounded below by $1/\mu$, $K^*(\rho)$ is piecewise linear and convex.
\end{proof}

The system-average waiting time is $W = \frac{1}{2} (g(\rho_1) + g(\rho_2))$, where $g(\rho) = K^*(\rho)/[2(1-\rho)]$.
The curvature of $g$ is:
\begin{equation}
g''(\rho) = \frac{K^{*\prime\prime}(\rho)(1-\rho)^2 + 2K^{*\prime}(\rho)(1-\rho) + 2K^*(\rho)}{2(1-\rho)^3}.
\end{equation}
For stable systems ($1-\rho > 0$), we have $K^{*\prime\prime}(\rho) \geq 0$, $K^{*\prime}(\rho) \geq 1/\mu > 0$ (Lemma~\ref{lem:frontier}), and $K^*(\rho) > 0$ for $\rho > 0$.
Thus $g''(\rho) > 0$, so $g$ is strictly convex.

Parametrize load imbalance by $\epsilon = \rho_1 - \bar{\rho}$, so $\rho_1 = \bar{\rho} + \epsilon$ and $\rho_2 = \bar{\rho} - \epsilon$ by work conservation.
The waiting time $W(\epsilon) = \frac{1}{2}[g(\bar{\rho}+\epsilon) + g(\bar{\rho}-\epsilon)]$ is symmetric in $\epsilon$ and strictly convex.
Therefore $W$ is uniquely minimized at $\epsilon = 0$, i.e., at balanced load $\rho_1 = \rho_2 = \bar{\rho}$.

\textbf{Step 4: Load-balanced policies are optimal for any fixed $p$.}
Fix a splitting probability $p$.
Among all policies with this $p$, non-load-balanced policies can induce load imbalance $\epsilon \neq 0$, while load-balanced policies achieve $\epsilon = 0$ (Step 2).
Since load imbalance strictly increases waiting time (Step 3), load-balanced policies achieve weakly lower waiting time.
\end{proof}

\begin{corollary}[Servers see identical arrivals]
\label{cor:servers_identical_arrivals}
Under load-balanced policies, both servers experience statistically identical $M^X/G/1$ arrival processes.
All analysis may therefore be conducted from the perspective of a single server without loss of generality.
\end{corollary}

This completes the technical foundations for the coordination challenge established in Section~\ref{sec:coordination} of the main text.

\section{Mapping to a Non-Local Game}
\label{sec:game-mapping}

In this appendix, we show that the distributed routing problem introduced in the main text and analyzed in detail in Appendix~\ref{sec:delta_threshold} can be formulated as a non-local game.
With this formulation, we establish an exact correspondence between performance in the non-local game and excess waiting time relative to the $w$-threshold policy.

\subsection{Local Strategies}

The analysis in Appendix~\ref{sec:delta_threshold} characterized the $w$-threshold policy, which is optimal when both service times $(X_1, X_2)$ are available for routing decisions.
This policy splits pairs with high splitting benefit $w(x_1, x_2)$ and bunches pairs with low splitting benefit, achieving the minimum waiting time at any given splitting probability.

In practice, each router observes only its local customer's service time.
Router~A (B) observes $x_1$ ($x_2$) and must choose an output based solely on this value.
This is precisely the setting of a non-local game, where non-communicating players receive local inputs and must produce outputs.

The key question becomes: how well can local strategies approximate the $w$-threshold policy?
Classical strategies (with shared randomness) and quantum strategies (with shared entanglement) may achieve different performance levels.
The non-local game formulation allows us to quantify this gap precisely.

\subsection{Setup}
Fix a splitting probability $p \in [0,1]$ and let $\tau_p$ be the unique threshold satisfying $\Pr[w(X_1,X_2) > \tau_p] = p$.
Define the decision function
\begin{equation}
\sigma^*(x_1,x_2) = \mathrm{sign}(\tau_p - w(x_1,x_2)),
\end{equation}
which takes value $+1$ (bunch) when $w < \tau_p$ and $-1$ (split) when $w > \tau_p$.
This is the $w$-threshold policy established in Appendix~\ref{sec:delta_threshold}: split when the splitting benefit exceeds the threshold, and bunch otherwise.

A local routing strategy is specified by output functions $o_A: \mathbb{R}_+ \to \{+1, -1\}$ and $o_B: \mathbb{R}_+ \to \{+1, -1\}$, where router~A observes $X_1$ and outputs $o_A(X_1)$, and router~B observes $X_2$ and outputs $o_B(X_2)$.
The product $o_A(X_1) o_B(X_2) = +1$ indicates bunching (both routers choose the same server) and $o_A(X_1) o_B(X_2) = -1$ indicates splitting (routers choose different servers).
We define the split indicator $s = (1 - o_A o_B)/2 \in \{0,1\}$, which equals $1$ when the pair is split and $0$ when bunched.
Note that $s \in \{0,1\}$ is a realized outcome for each pair, whereas $r(x_1, x_2) \in [0,1]$ in Appendix~\ref{sec:delta_threshold} denotes a splitting probability.

We restrict attention to strategies achieving splitting probability $p$:
\begin{equation}
\Pr[o_A(X_1) o_B(X_2) = -1] = p.
\end{equation}
This constraint is natural because, as established in Appendix~\ref{sec:baseline_throughput}, baseline throughput depends only on the splitting probability $p$.
Comparing strategies at fixed $p$ therefore isolates the waiting time objective.

\begin{remark}[Load-balanced server assignment]
\label{rem:load-balanced}
The waiting time analysis in Appendix~\ref{sec:delta_threshold} assumes load-balanced server assignment: when splitting, each server is equally likely to receive either job; when bunching, the destination server is chosen uniformly at random.
For local strategies, load balancing is achievable via shared randomness: players share a random bit $b \in \{0,1\}$, and if $b = 1$, both flip their outputs (i.e., output $-o_A(X_1)$ and $-o_B(X_2)$ instead of $o_A(X_1)$ and $o_B(X_2)$).
This transformation preserves the product $o_A o_B$ (hence the split/bunch decision and game payoff) while ensuring symmetric server assignment.
Throughout this appendix, we assume all strategies employ load-balanced server assignment, so that the waiting time results from Appendix~\ref{sec:delta_threshold} apply.
\end{remark}

\subsection{The Queueing--Game Correspondence}

\begin{lemma}[Waiting time gap]
\label{lem:wt-gap}
Let $(o_A, o_B)$ be any local routing strategy with splitting probability $p$, employing load-balanced server assignment.
The excess waiting time relative to the $w$-threshold policy is
\begin{equation}
\label{eq:delta-wq}
\Delta W_q := \mathbb{E}[W_q](o_A, o_B) - \mathbb{E}[W_q]^* = \frac{1}{2}\,\mathbb{E}\bigl[(o_A(X_1) o_B(X_2) - \sigma^*(X_1,X_2))\, w(X_1,X_2)\bigr].
\end{equation}
\end{lemma}

\begin{proof}
From Appendix~\ref{sec:form_fixed_prob}, under load-balanced server assignment the expected waiting time under any routing policy with split indicator $s: \mathbb{R}_+^2 \to \{0,1\}$ is
\begin{equation}
\mathbb{E}[W_q] = C - \mathbb{E}[s \cdot w],
\end{equation}
where $C = \frac{\lambda \mathbb{E}[S^2]}{4(1-\rho)} + \frac{\mathbb{E}[S]}{4}$ is independent of the routing policy and $w(x_1,x_2) = c_1 x_1 x_2 + c_2(x_1+x_2)$ is the splitting benefit function.

The $w$-threshold policy achieves $s^*(x_1,x_2) = \mathbf{1}\{w(x_1,x_2) > \tau_p\}$, which can be written as
\begin{equation}
s^* = \frac{1 - \sigma^*}{2}
\end{equation}
since $\sigma^* = +1$ when $w < \tau_p$ (bunch, $s^* = 0$) and $\sigma^* = -1$ when $w > \tau_p$ (split, $s^* = 1$).

Similarly, for a local strategy $(o_A, o_B)$ with split indicator $s = (1 - o_A o_B)/2$, we have
\begin{equation}
s^* - s = \frac{1 - \sigma^*}{2} - \frac{1 - o_A o_B}{2} = \frac{o_A o_B - \sigma^*}{2}.
\end{equation}

The waiting time gap is therefore
\begin{align}
\Delta W_q &= \mathbb{E}[W_q](o_A, o_B) - \mathbb{E}[W_q]^* \\
&= \bigl(C - \mathbb{E}[s \cdot w]\bigr) - \bigl(C - \mathbb{E}[s^* w]\bigr) \\
&= \mathbb{E}[s^* w] - \mathbb{E}[s \cdot w] \\
&= \mathbb{E}[(s^* - s) \cdot w] \\
&= \frac{1}{2}\,\mathbb{E}[(o_A o_B - \sigma^*) \cdot w]. \qedhere
\end{align}
\end{proof}

\begin{remark}[Interpretation]
The waiting time gap decomposes into contributions from two types of disagreement with $\sigma^*$:
\begin{itemize}
\item \emph{Under-splitting}: bunching when $\sigma^* = -1$ (i.e., $o_A o_B = +1$), contributing $+w$ to $\Delta W_q$.
\item \emph{Over-splitting}: splitting when $\sigma^* = +1$ (i.e., $o_A o_B = -1$), contributing $-w$ to $\Delta W_q$.
\end{itemize}
Since strategies are constrained to the same splitting probability $p$, any over-splitting must be compensated by under-splitting elsewhere.
However, under-splitting occurs in the high-$w$ region (above threshold) while over-splitting occurs in the low-$w$ region (below threshold), so under-splitting is more costly.
This asymmetry ensures $\Delta W_q \geq 0$ with equality if and only if $o_A o_B = \sigma^*$ almost surely.
\end{remark}

\subsection{Non-Local Game Formulation}

We cast the routing problem as a non-local game (Definition~\ref{def:nonlocal-game}).
Players A and B receive service times $X_1, X_2 \sim \mathrm{Exp}(\mu)$ independently and produce outputs $o_A(X_1), o_B(X_2) \in \{+1, -1\}$ without communication.
The game payoff is
\begin{equation}
\label{eq:game-payoff}
A(o_A, o_B) = -\mathbb{E}[o_A(X_1) o_B(X_2) \cdot w(X_1,X_2)],
\end{equation}
subject to the constraint $\Pr[o_A(X_1) o_B(X_2) = -1] = p$.
The payoff $A$ rewards splitting (which corresponds to $o_A o_B = -1$) in proportion to the splitting benefit $w$, directly encoding the queueing objective $\mathbb{E}[s \cdot w]$ that determines waiting time.

\begin{proposition}[Game--queueing equivalence; restates Lemma~\ref{lem:agreement_queue_main}]
\label{prop:game-queueing}
Let $A^* = -\mathbb{E}[\sigma^* \cdot w]$ denote the $w$-threshold payoff.
For any local strategy $(o_A, o_B)$ with splitting probability $p$ and load-balanced server assignment:
\begin{equation}
\Delta W_q = \frac{A^* - A(o_A, o_B)}{2}.
\end{equation}
Consequently, maximizing the game payoff $A(o_A, o_B)$ is equivalent to minimizing the excess waiting time $\Delta W_q$.
\end{proposition}

\begin{proof}
From Lemma~\ref{lem:wt-gap}:
\begin{align}
\Delta W_q &= \frac{1}{2}\,\mathbb{E}[(o_A o_B - \sigma^*) \cdot w] \\
&= \frac{1}{2}\bigl(\mathbb{E}[o_A o_B \cdot w] - \mathbb{E}[\sigma^* \cdot w]\bigr) \\
&= \frac{1}{2}\bigl(-A(o_A, o_B) - (-A^*)\bigr) \\
&= \frac{A^* - A(o_A, o_B)}{2}. \qedhere
\end{align}
\end{proof}

\begin{corollary}[Quantum advantage implies waiting time reduction]
\label{cor:quantum-advantage}
Let $A_{\mathrm{cl}}^*(p)$ denote the maximum payoff achievable by a classical strategy at splitting probability $p$, and let $A_{\mathrm{qu}}(p)$ denote the payoff achievable by a quantum strategy at the same splitting probability.
If $A_{\mathrm{qu}}(p) > A_{\mathrm{cl}}^*(p)$, then
\begin{equation}
\Delta W_q^{\mathrm{qu}}(p) < \Delta W_q^{\mathrm{cl}}(p),
\end{equation}
i.e., the quantum strategy achieves strictly lower excess waiting time at splitting probability $p$.
\end{corollary}

The following theorem extends this pointwise comparison to Pareto dominance across a range of splitting probabilities.

\begin{theorem}[Pareto dominance; restates Theorem~\ref{thm:main}]
\label{thm:pareto-dominance}
Let $\mathcal{P} \subseteq [0,1]$ be the set of splitting probabilities for which $A_{\mathrm{qu}}(p) > A_{\mathrm{cl}}^*(p)$.
Then quantum strategies Pareto dominate classical strategies over $\mathcal{P}$: for all $p \in \mathcal{P}$,
\begin{equation}
\Delta W_q^{\mathrm{qu}}(p) < \Delta W_q^{\mathrm{cl}}(p)
\end{equation}
at identical baseline throughput $\mathcal{T}(p)$.
\end{theorem}

\begin{proof}
Fix $p \in \mathcal{P}$.
By Proposition~\ref{prop:game-queueing}, the excess waiting times satisfy
\begin{equation}
\Delta W_q^{\mathrm{qu}}(p) = \frac{A^* - A_{\mathrm{qu}}(p)}{2} < \frac{A^* - A_{\mathrm{cl}}^*(p)}{2} = \Delta W_q^{\mathrm{cl}}(p),
\end{equation}
where the inequality follows from $A_{\mathrm{qu}}(p) > A_{\mathrm{cl}}^*(p)$.

By Proposition~\ref{prop:baseline-monotonicity} in Appendix~\ref{sec:baseline_throughput}, baseline throughput depends only on the splitting probability $p$, not on which specific pairs are split.
Therefore both strategies achieve the same baseline throughput $\mathcal{T}(p)$.

Since quantum strategies achieve strictly lower waiting time at identical throughput for all $p \in \mathcal{P}$, they Pareto dominate classical strategies over this range.
\end{proof}

\begin{remark}
Theorem~\ref{thm:pareto-dominance} establishes that demonstrating quantum advantage in the non-local game immediately implies Pareto dominance in the waiting time--throughput trade-off.
It remains to characterize what classical strategies can achieve (Appendix~\ref{app:classical_strategies}) and compute classical and quantum performance numerically (Appendix~\ref{sec:numerical_methods}).
\end{remark}

\section{Classical Strategies}\label{app:classical_strategies}

This appendix characterizes what classical strategies can achieve in the routing game defined in Appendix~\ref{sec:game-mapping}.
We first show that optimal deterministic strategies take a threshold form, then establish the correct classical benchmark when shared randomness is allowed.

\subsection{Optimal Deterministic Strategies}
\label{sec:classical-threshold}

We show that optimal deterministic strategies take a simple threshold form, reducing the infinite-dimensional optimization to a two-dimensional problem.

\begin{definition}[Threshold strategy]
\label{def:threshold}
A strategy $f: \mathbb{R}_+ \to \{+1, -1\}$ is a \emph{threshold strategy} with threshold $\theta \geq 0$ if
\begin{equation}
f(x) = \begin{cases}
+1 & x < \theta, \\
-1 & x \geq \theta.
\end{cases}
\end{equation}
We write $f_\theta$ to denote the threshold strategy with threshold $\theta$.
\end{definition}

The main result of this section is the following.

\begin{theorem}[Optimal deterministic strategies are threshold strategies; restates Theorem~\ref{thm:classical-threshold-main}]
\label{thm:classical-threshold}
For the routing game at any splitting probability $p \in (0,1)$, there exists an optimal deterministic strategy $(o_A^*, o_B^*)$ where both $o_A^*$ and $o_B^*$ are threshold strategies.
More precisely, the optimal strategy takes the form of a threshold pair up to a global sign flip $(o_A, o_B) \mapsto (-o_A, -o_B)$, which preserves the payoff and splitting constraint; we may therefore always choose the standard orientation of Definition~\ref{def:threshold}.
\end{theorem}

The proof proceeds by showing that, for any fixed strategy of one player, the optimal response of the other player is a threshold strategy.

\begin{lemma}[Optimal response is a threshold]
\label{lem:optimal-response}
Fix a strategy $o_B: \mathbb{R}_+ \to \{+1, -1\}$ with $E_0 := \mathbb{E}[o_B(X)] \in (-1, 1)$, $E_0 \neq 0$, and $E_1 := \mathbb{E}[o_B(X) \cdot X]$.
Among all strategies $o_A$ satisfying the splitting constraint $\Pr[o_A(X_1) o_B(X_2) = -1] = p$, the payoff $A(o_A, o_B)$ is uniquely maximized by a threshold strategy.
The case $E_0 = 0$ is handled in Remark~\ref{rem:E0-zero}.
\end{lemma}

\begin{proof}
We first rewrite the payoff in terms of $o_A$ alone.
Using independence of $X_1$ and $X_2$:
\begin{align}
A(o_A, o_B) &= -\mathbb{E}[o_A(X_1) o_B(X_2) \cdot w(X_1, X_2)] \\
&= -\mathbb{E}_{X_1}\bigl[o_A(X_1) \cdot \mathbb{E}_{X_2}[o_B(X_2) \cdot w(X_1, X_2)]\bigr].
\end{align}

Computing the inner expectation with $w(x_1, x_2) = c_1 x_1 x_2 + c_2(x_1 + x_2)$:
\begin{align}
\mathbb{E}_{X_2}[o_B(X_2) \cdot w(x_1, X_2)] &= c_1 x_1 \mathbb{E}[o_B(X) \cdot X] + c_2 x_1 \mathbb{E}[o_B(X)] + c_2 \mathbb{E}[o_B(X) \cdot X] \\
&= (c_1 E_1 + c_2 E_0)\, x_1 + c_2 E_1.
\end{align}

Define the effective weight function $\tilde{w}: \mathbb{R}_+ \to \mathbb{R}$ by
\begin{equation}
\tilde{w}(x) := (c_1 E_1 + c_2 E_0)\, x + c_2 E_1.
\end{equation}
Then the payoff becomes
\begin{equation}
A(o_A, o_B) = -\mathbb{E}[o_A(X) \cdot \tilde{w}(X)].
\end{equation}

The splitting constraint $\Pr[o_A(X_1) o_B(X_2) = -1] = p$ can be rewritten using independence:
\begin{equation}
p = \Pr[o_A o_B = -1] = \frac{1 - \mathbb{E}[o_A(X)]\mathbb{E}[o_B(X)]}{2} = \frac{1 - D_0 E_0}{2},
\end{equation}
where $D_0 := \mathbb{E}[o_A(X)]$.
This yields the constraint
\begin{equation}
D_0 = \frac{1 - 2p}{E_0}.
\end{equation}

The optimization over $o_A$ is therefore:
\begin{equation}
\begin{aligned}
\text{maximize} \quad & -\mathbb{E}[o_A(X) \cdot \tilde{w}(X)] \\
\text{subject to} \quad & \mathbb{E}[o_A(X)] = D_0^{\mathrm{target}} := \frac{1 - 2p}{E_0}, \\
& o_A(x) \in \{+1, -1\} \text{ for all } x \geq 0.
\end{aligned}
\end{equation}

This is a linear functional optimization with a linear constraint.
We solve it via Lagrangian relaxation.
The Lagrangian is
\begin{equation}
\mathcal{L}(o_A, \lambda) = -\mathbb{E}[o_A(X) \cdot \tilde{w}(X)] - \lambda\bigl(\mathbb{E}[o_A(X)] - D_0^{\mathrm{target}}\bigr) = -\mathbb{E}[o_A(X) \cdot (\tilde{w}(X) + \lambda)] + \lambda D_0^{\mathrm{target}}.
\end{equation}

For any fixed $\lambda$, the Lagrangian is maximized pointwise by choosing
\begin{equation}
o_A(x) = \begin{cases}
+1 & \text{if } \tilde{w}(x) + \lambda < 0, \\
-1 & \text{if } \tilde{w}(x) + \lambda > 0.
\end{cases}
\end{equation}

Since $\tilde{w}(x) = (c_1 E_1 + c_2 E_0)\, x + c_2 E_1$ is affine in $x$, the condition $\tilde{w}(x) + \lambda \lessgtr 0$ defines a half-line whenever the slope $c_1 E_1 + c_2 E_0$ is nonzero.

If $c_1 E_1 + c_2 E_0 > 0$, then $\tilde{w}$ is strictly increasing, so $\tilde{w}(x) + \lambda < 0$ if and only if $x < \theta$ for some threshold $\theta$.
The optimal response satisfies $o_A(x) = +1$ for $x < \theta$ and $o_A(x) = -1$ for $x \geq \theta$.

If $c_1 E_1 + c_2 E_0 < 0$, then $\tilde{w}$ is strictly decreasing, so $\tilde{w}(x) + \lambda < 0$ if and only if $x > \theta$ for some threshold $\theta$.
The optimal response satisfies $o_A(x) = -1$ for $x < \theta$ and $o_A(x) = +1$ for $x \geq \theta$, which is the negation of a threshold strategy: $o_A = -f_\theta$.

In both cases, the optimal $o_A$ is a threshold strategy up to a global sign flip.
Since the game payoff depends on the product $o_A o_B$, replacing $(o_A, o_B)$ with $(-o_A, -o_B)$ leaves the payoff and splitting constraint unchanged.
Thus any optimum with reversed-orientation strategies corresponds to an equivalent optimum with standard-orientation strategies.

In the degenerate case $c_1 E_1 + c_2 E_0 = 0$, the effective weight is constant and the optimal response is $o_A \equiv \pm 1$, which can be viewed as a limiting threshold strategy with $\theta \in \{0, \infty\}$.
For the numerical optimization, we verify that optimal threshold pairs $(\theta_A, \theta_B)$ lie in the interior of the feasible region where the non-degeneracy condition holds.

The threshold $\theta$ is determined by the constraint $\mathbb{E}[o_A(X)] = D_0^{\mathrm{target}}$; since $\mathbb{E}[f_\theta(X)]$ is continuous and strictly monotonic in $\theta$, a unique such $\theta$ exists for any feasible $D_0^{\mathrm{target}} \in (-1, 1)$.
\end{proof}

\begin{remark}[The case $E_0 = 0$]
\label{rem:E0-zero}
When $E_0 = 0$, the splitting constraint $D_0 E_0 = 1 - 2p$ forces $p = 1/2$, and the constraint on $D_0 = \mathbb{E}[o_A(X)]$ becomes vacuous.
The optimization over $o_A$ reduces to the unconstrained pointwise problem $\max_{o_A} -\mathbb{E}[o_A(X) \cdot \tilde{w}(X)]$, solved by $o_A(x) = -\operatorname{sign}(\tilde{w}(x))$.
With $E_0 = 0$, the effective weight simplifies to $\tilde{w}(x) = E_1(c_1 x + c_2)$, which has constant sign on $\mathbb{R}_+$ (since $c_1, c_2 > 0$).
The optimal response is therefore a constant function $o_A \equiv -\operatorname{sign}(E_1)$, corresponding to a limiting threshold strategy with $\theta = 0$ or $\theta = \infty$.
In the fully degenerate case $E_0 = E_1 = 0$, $\tilde{w} \equiv 0$ and any strategy $o_A$ is optimal; a threshold strategy may again be chosen.
Thus the conclusion of Lemma~\ref{lem:optimal-response} extends to $E_0 = 0$.
\end{remark}

\begin{proof}[Proof of Theorem~\ref{thm:classical-threshold}]
Suppose $(o_A^*, o_B^*)$ is a globally optimal deterministic strategy.
We show both $o_A^*$ and $o_B^*$ must be threshold strategies (up to global sign flip).

Suppose $o_A^*$ is not a threshold strategy (in either orientation).
By Lemma~\ref{lem:optimal-response} (or Remark~\ref{rem:E0-zero} when $E_0 = 0$), there exists a threshold strategy $\tilde{o}_A$ (up to sign flip) such that $A(\tilde{o}_A, o_B^*) > A(o_A^*, o_B^*)$, contradicting global optimality of $(o_A^*, o_B^*)$.
Therefore $o_A^*$ is a threshold strategy up to sign flip.

By an identical argument with the roles of $o_A$ and $o_B$ exchanged (the problem is symmetric under player exchange), $o_B^*$ is also a threshold strategy up to sign flip.
By the $(o_A, o_B) \mapsto (-o_A, -o_B)$ symmetry, the pair can always be taken to have standard orientation.
\end{proof}

\begin{remark}[Reduction to finite-dimensional optimization]
Theorem~\ref{thm:classical-threshold} reduces the search for optimal deterministic strategies from an infinite-dimensional space of function pairs $(o_A, o_B)$ to a two-dimensional optimization over threshold pairs $(\theta_A, \theta_B) \in \mathbb{R}_+^2$.
The explicit parameterization and certified computation procedure are given in Appendix~\ref{sec:numerical_methods}.
\end{remark}

\subsection{Shared Randomness and the Classical Benchmark}
\label{sec:shared-randomness}

We now consider classical strategies augmented with shared randomness.
A shared-randomness strategy uses a common random variable $\kappa$ (independent of the inputs $X_1, X_2$) to select a deterministic local strategy $(o_{A,\kappa}, o_{B,\kappa})$ for each arriving pair.
The resulting payoff is $\mathbb{E}_\kappa[A(o_{A,\kappa}, o_{B,\kappa})]$, and the splitting constraint requires $\mathbb{E}_\kappa[p_\kappa] = p$, where $p_\kappa = \Pr[o_{A,\kappa}(X_1) o_{B,\kappa}(X_2) = -1]$.

\begin{proposition}[Fixed-$p$ mixtures]
\label{prop:fixed-p-mixture}
Since the payoff $A(o_A, o_B)$ is linear in the strategy, shared randomness provides no advantage at fixed splitting probability: any convex mixture of deterministic strategies all achieving splitting probability $p$ has payoff at most $A_{\mathrm{cl}}^*(p)$.
\end{proposition}

\begin{proof}
Let $(o_{A,i}, o_{B,i})$ for $i = 1, \ldots, n$ be deterministic strategies each achieving splitting probability $p$, and let $q_i \geq 0$ with $\sum_i q_i = 1$.
The mixture achieves payoff
\begin{equation}
A_{\mathrm{SR}} = \sum_i q_i A(o_{A,i}, o_{B,i}) \leq \max_i A(o_{A,i}, o_{B,i}) \leq A_{\mathrm{cl}}^*(p). \qedhere
\end{equation}
\end{proof}

However, shared randomness \emph{can} improve performance by mixing strategies at \emph{different} splitting probabilities $p_1, \ldots, p_n$ with $\sum_i q_i p_i = p$.
This motivates the following definition.

\begin{definition}[Shared-randomness classical value]
The shared-randomness classical value at splitting probability $p$ is
\begin{equation}
A_{\mathrm{cl,SR}}^*(p) := \sup \bigl\{ \mathbb{E}_\kappa[A(o_{A,\kappa}, o_{B,\kappa})] : \mathbb{E}_\kappa[p_\kappa] = p \bigr\},
\end{equation}
the supremum over all shared-randomness strategies with mean splitting probability $p$.
\end{definition}

\begin{proposition}[Concavity]
\label{prop:concavity}
The shared-randomness classical value $A_{\mathrm{cl,SR}}^*(p)$ is concave in $p$.
\end{proposition}

\begin{proof}
Fix $p_1, p_2 \in [0,1]$, $\alpha \in (0,1)$, and let $p = \alpha p_1 + (1-\alpha) p_2$.
For any $\epsilon > 0$, let $S_1, S_2$ be strategies achieving payoffs within $\epsilon$ of $A_{\mathrm{cl,SR}}^*(p_1)$ and $A_{\mathrm{cl,SR}}^*(p_2)$ respectively.
Consider the strategy that selects $S_1$ with probability $\alpha$ and $S_2$ with probability $1-\alpha$ via an independent Bernoulli$(\alpha)$ coin flip.
This strategy has mean splitting probability $\alpha p_1 + (1-\alpha) p_2 = p$ and achieves expected payoff at least $\alpha (A_{\mathrm{cl,SR}}^*(p_1) - \epsilon) + (1-\alpha)(A_{\mathrm{cl,SR}}^*(p_2) - \epsilon)$.
Taking $\epsilon \to 0$ yields $A_{\mathrm{cl,SR}}^*(p) \geq \alpha A_{\mathrm{cl,SR}}^*(p_1) + (1-\alpha) A_{\mathrm{cl,SR}}^*(p_2)$.
\end{proof}

The key result is that $A_{\mathrm{cl,SR}}^*(p)$ equals the concave envelope of the deterministic value.

\begin{theorem}[Concavification]
\label{thm:concavification}
The shared-randomness classical value equals the concave envelope of the deterministic classical value:
\begin{equation}
A_{\mathrm{cl,SR}}^*(p) = \mathrm{conc}(A_{\mathrm{cl}}^*)(p),
\end{equation}
where $\mathrm{conc}(f)$ denotes the upper concave envelope of $f$.
\end{theorem}

\begin{proof}
\emph{Upper bound.}
Any shared-randomness strategy is a convex mixture of deterministic strategies.
Each deterministic component with splitting probability $p_i$ achieves payoff at most $A_{\mathrm{cl}}^*(p_i)$.
The mixture therefore achieves payoff at most the corresponding convex combination of $A_{\mathrm{cl}}^*(p_i)$ values, which lies on or below the concave envelope.

\emph{Lower bound.}
By Proposition~\ref{prop:concavity}, $A_{\mathrm{cl,SR}}^*$ is concave.
Since deterministic strategies are a special case of shared-randomness strategies, $A_{\mathrm{cl,SR}}^*(p) \geq A_{\mathrm{cl}}^*(p)$ for all $p$.
A concave function that dominates $A_{\mathrm{cl}}^*$ pointwise must dominate $\mathrm{conc}(A_{\mathrm{cl}}^*)$.
\end{proof}

\begin{remark}[The correct classical benchmark]
\label{rem:classical-benchmark}
For certifying quantum advantage, the correct classical benchmark is $A_{\mathrm{cl,SR}}^*(p)$, not $A_{\mathrm{cl}}^*(p)$.
If $A_{\mathrm{cl}}^*(p)$ happens to be concave (so the concave envelope coincides with the original function), then shared randomness provides no advantage at any splitting probability.
\end{remark}

This completes the characterization of classical strategies.
The numerical methods for computing $A_{\mathrm{cl}}^*(p)$ and its concave envelope are presented in Appendix~\ref{sec:numerical_methods}.

\section{Numerical Methods}
\label{sec:numerical_methods}

This appendix describes the numerical methods used to compute the waiting time--throughput trade-off curves for classical and quantum strategies presented in the main text.

\subsection{Overview}

To construct the Pareto frontier comparison between classical and quantum strategies, we compute three quantities at each splitting probability $p \in (0, 1/2)$:
\begin{enumerate}
    \item The $w$-threshold payoff $A^*(p)$, achieved by the optimal policy with full access to both service times;
    \item The optimal classical payoff $A_{\mathrm{cl}}^*(p)$, achieved by the best local threshold strategy;
    \item The quantum payoff $A_{\mathrm{qu}}(p)$, achieved by a numerically optimized quantum strategy.
\end{enumerate}
The waiting time gap relative to the $w$-threshold policy is then obtained via the game--queueing correspondence (Proposition~\ref{prop:game-queueing}):
\begin{equation}
    \Delta W_q = \frac{A^*(p) - A(p)}{2},
\end{equation}
where $A(p)$ is the payoff achieved by either the classical or quantum strategy.

\subsection{Classical Upper Bound via Threshold Strategies}

By Theorem~\ref{thm:classical-threshold}, optimal deterministic classical strategies take the form of threshold strategies.
This section provides the explicit parameterization and certified computation procedure.

\subsubsection{Explicit Parameterization}

For $X \sim \mathrm{Exp}(\mu)$ and threshold strategy $f_\theta$, the relevant moments are:
\begin{align}
D_0(\theta) &:= \mathbb{E}[f_\theta(X)] = \Pr[X < \theta] - \Pr[X \geq \theta] = 1 - 2e^{-\mu\theta}, \\
D_1(\theta) &:= \mathbb{E}[f_\theta(X) \cdot X] = \int_0^\theta x \mu e^{-\mu x}\, dx - \int_\theta^\infty x \mu e^{-\mu x}\, dx.
\end{align}

Evaluating the integrals:
\begin{align}
\int_0^\theta x \mu e^{-\mu x}\, dx &= \frac{1}{\mu}\bigl(1 - e^{-\mu\theta}(1 + \mu\theta)\bigr), \\
\int_\theta^\infty x \mu e^{-\mu x}\, dx &= \frac{1}{\mu}e^{-\mu\theta}(1 + \mu\theta).
\end{align}

Therefore:
\begin{equation}
D_1(\theta) = \frac{1}{\mu}\bigl(1 - 2e^{-\mu\theta}(1 + \mu\theta)\bigr) = \frac{1}{\mu}D_0(\theta) - 2\theta e^{-\mu\theta}.
\end{equation}

For threshold strategies with thresholds $\theta_A$ and $\theta_B$, the game payoff becomes:
\begin{equation}
\label{eq:payoff-thresholds}
A(\theta_A, \theta_B) = -c_1 D_1(\theta_A) D_1(\theta_B) - c_2\bigl(D_1(\theta_A) D_0(\theta_B) + D_0(\theta_A) D_1(\theta_B)\bigr),
\end{equation}
and the splitting constraint becomes:
\begin{equation}
\label{eq:constraint-thresholds}
D_0(\theta_A) \cdot D_0(\theta_B) = 1 - 2p.
\end{equation}

The optimal classical value at splitting probability $p$ is:
\begin{equation}
A_{\mathrm{cl}}^*(p) = \max_{\substack{\theta_A, \theta_B \geq 0 \\ D_0(\theta_A) D_0(\theta_B) = 1 - 2p}} A(\theta_A, \theta_B).
\end{equation}

\subsubsection{Reduction to One-Dimensional Optimization}

The constraint $D_0(\theta_A) \cdot D_0(\theta_B) = 1 - 2p$ defines a one-dimensional curve in $(\theta_A, \theta_B)$ space.
We parameterize this curve by $\theta_A$ and solve for $\theta_B$.

Recall that $D_0(\theta) = 1 - 2e^{-\mu\theta}$ is strictly increasing in $\theta$, with $D_0(0) = -1$ and $\lim_{\theta \to \infty} D_0(\theta) = 1$.
Given $\theta_A$, the constraint requires
\begin{equation}
D_0(\theta_B) = \frac{1 - 2p}{D_0(\theta_A)}.
\end{equation}
For $\theta_B$ to exist (i.e., for the right-hand side to lie in $(-1, 1)$), we need
\begin{equation}
-1 < \frac{1 - 2p}{D_0(\theta_A)} < 1.
\end{equation}

For $p \in (0, 1/2)$, we have $1 - 2p > 0$.
The constraint $\frac{1-2p}{D_0(\theta_A)} < 1$ requires $D_0(\theta_A) > 1 - 2p$, i.e., $\theta_A > \theta_{\min}$ where
\begin{equation}
\theta_{\min} = -\frac{\ln p}{\mu}.
\end{equation}
The constraint $\frac{1-2p}{D_0(\theta_A)} > -1$ imposes no additional restriction beyond $\theta_A > \theta_{\min}$: once $D_0(\theta_A) > 1-2p > 0$, the lower bound is automatically satisfied since $\frac{1-2p}{D_0(\theta_A)} > 0 > -1$.

For $p \in (1/2, 1)$, we have $1 - 2p < 0$, and a symmetric analysis applies with the roles of upper and lower bounds exchanged.

Henceforth, we focus on $p \in (0, 1/2)$; the case $p > 1/2$ is analogous.
The feasible range is $\theta_A \in (\theta_{\min}, \infty)$, and by symmetry of the problem, we may restrict to $\theta_A \in (\theta_{\min}, \theta_{\mathrm{sym}}]$ where $\theta_{\mathrm{sym}}$ is the symmetric point satisfying $\theta_A = \theta_B$.

Given feasible $\theta_A$, we compute
\begin{equation}
\theta_B(\theta_A) = -\frac{1}{\mu} \ln\left(\frac{1}{2}\left(1 - \frac{1 - 2p}{D_0(\theta_A)}\right)\right).
\end{equation}

Define the reduced objective function $\tilde{A}: (\theta_{\min}, \infty) \to \mathbb{R}$ by
\begin{equation}
\tilde{A}(\theta_A) := A(\theta_A, \theta_B(\theta_A)).
\end{equation}

The optimal classical value is
\begin{equation}
A_{\mathrm{cl}}^*(p) = \sup_{\theta_A > \theta_{\min}} \tilde{A}(\theta_A).
\end{equation}

\subsubsection{Certified Grid Search}

To compute $A_{\mathrm{cl}}^*(p)$ with certified error bounds, we discretize the feasible interval and bound the error due to discretization.

Fix a compact search interval $[\theta_{\min} + \epsilon, \theta_{\max}]$ for some small $\epsilon > 0$ and large $\theta_{\max}$.
We separately bound the contribution from the excluded regions $(\theta_{\min}, \theta_{\min} + \epsilon)$ and $(\theta_{\max}, \infty)$.

\begin{definition}[Grid approximation]
For grid spacing $\delta > 0$, define the grid points
\begin{equation}
\theta_A^{(k)} = \theta_{\min} + \epsilon + k\delta, \quad k = 0, 1, \ldots, K,
\end{equation}
where $K = \lfloor (\theta_{\max} - \theta_{\min} - \epsilon) / \delta \rfloor$.
The grid approximation to the optimal value is
\begin{equation}
A_{\mathrm{grid}} := \max_{k \in \{0, \ldots, K\}} \tilde{A}(\theta_A^{(k)}).
\end{equation}
\end{definition}

\paragraph{Lipschitz Bound.}

To control the discretization error, we bound the derivative of $\tilde{A}$.

\begin{lemma}[Lipschitz bound]
\label{lem:lipschitz}
On any compact interval $[\theta_a, \theta_b] \subset (\theta_{\min}, \infty)$, the reduced objective $\tilde{A}$ is Lipschitz continuous with constant
\begin{equation}
L = \sup_{\theta_A \in [\theta_a, \theta_b]} \left| \frac{d\tilde{A}}{d\theta_A} \right|.
\end{equation}
\end{lemma}

\begin{proof}
The functions $D_0(\theta)$ and $D_1(\theta)$ are smooth (in fact, analytic) on $(0, \infty)$.
The constraint mapping $\theta_A \mapsto \theta_B(\theta_A)$ is smooth on $(\theta_{\min}, \infty)$.
The payoff $A(\theta_A, \theta_B)$ is a polynomial in $D_0, D_1$ values, hence smooth.
By the chain rule, $\tilde{A}(\theta_A) = A(\theta_A, \theta_B(\theta_A))$ is smooth, and in particular Lipschitz on compact subintervals.
\end{proof}

To compute $L$ explicitly, we use the chain rule:
\begin{equation}
\frac{d\tilde{A}}{d\theta_A} = \frac{\partial A}{\partial \theta_A} + \frac{\partial A}{\partial \theta_B} \cdot \frac{d\theta_B}{d\theta_A}.
\end{equation}

The partial derivatives are:
\begin{align}
\frac{\partial A}{\partial \theta_A} &= -c_1 D_1'(\theta_A) D_1(\theta_B) - c_2 \bigl( D_1'(\theta_A) D_0(\theta_B) + D_0'(\theta_A) D_1(\theta_B) \bigr), \\
\frac{\partial A}{\partial \theta_B} &= -c_1 D_1(\theta_A) D_1'(\theta_B) - c_2 \bigl( D_1(\theta_A) D_0'(\theta_B) + D_0(\theta_A) D_1'(\theta_B) \bigr),
\end{align}
where
\begin{align}
D_0'(\theta) &= 2\mu e^{-\mu\theta}, \\
D_1'(\theta) &= 2\mu \theta e^{-\mu\theta}.
\end{align}

The derivative of the constraint mapping is obtained by implicit differentiation of $D_0(\theta_A) D_0(\theta_B) = 1 - 2p$:
\begin{equation}
\frac{d\theta_B}{d\theta_A} = -\frac{D_0'(\theta_A) D_0(\theta_B)}{D_0(\theta_A) D_0'(\theta_B)} = -\frac{D_0(\theta_B)}{D_0(\theta_A)} \cdot \frac{e^{-\mu\theta_A}}{e^{-\mu\theta_B}}.
\end{equation}

On a compact interval $[\theta_a, \theta_b]$, all quantities $D_0, D_1, D_0', D_1', \theta_B, \frac{d\theta_B}{d\theta_A}$ are bounded.
The Lipschitz constant $L$ can be computed numerically by evaluating $|d\tilde{A}/d\theta_A|$ on a fine grid and taking the maximum, or bounded analytically by bounding each factor.

\paragraph{Boundary Contributions.}

We must also bound the objective on the excluded regions.

\begin{lemma}[Boundary bounds]
\label{lem:boundary}
The reduced objective $\tilde{A}(\theta_A)$ extends continuously to finite limits at the boundaries:
\begin{align}
\lim_{\theta_A \to \theta_{\min}^+} \tilde{A}(\theta_A) &= A_{\min}(p), \\
\lim_{\theta_A \to \infty} \tilde{A}(\theta_A) &= A_{\infty}(p),
\end{align}
where $A_{\min}(p)$ and $A_{\infty}(p)$ are finite values.
\end{lemma}

\begin{proof}
As $\theta_A \to \theta_{\min}^+$, we have $D_0(\theta_A) \to (1 - 2p)^+$, so $D_0(\theta_B) \to 1^-$, which requires $\theta_B \to \infty$.
In this limit:
\begin{itemize}
\item $D_0(\theta_A) \to 1 - 2p$ (finite),
\item $D_1(\theta_A) \to D_1(\theta_{\min})$ (finite),
\item $D_0(\theta_B) \to 1$ (finite),
\item $D_1(\theta_B) \to 1/\mu$ (finite).
\end{itemize}
Since the payoff $A(\theta_A, \theta_B)$ is a polynomial in $D_0, D_1$ values, the limit $A_{\min}(p)$ is finite.

As $\theta_A \to \infty$, we have $D_0(\theta_A) \to 1$, so $D_0(\theta_B) \to 1 - 2p$, giving $\theta_B \to \theta_{\min}$.
By a symmetric argument, all quantities remain finite and $A_{\infty}(p)$ is finite.

By symmetry of the problem under $\theta_A \leftrightarrow \theta_B$, we have $A_{\min}(p) = A_{\infty}(p)$.
\end{proof}

\begin{remark}
Although the boundary limits are finite, numerical evaluation confirms that for $p$ in the range of interest, these limits lie below the interior maximum.
The derivative $|d\tilde{A}/d\theta_A|$ diverges as $\theta_A \to \theta_{\min}^+$ (since $d\theta_B/d\theta_A$ grows like $e^{\mu\theta_B}$), indicating that the payoff changes rapidly near the boundary, but this does not imply the limit is $-\infty$.
\end{remark}

\paragraph{Certified Upper Bound.}

Combining the above, we obtain:

\begin{proposition}[Certified classical bound]
\label{prop:certified-bound}
Let $[\theta_{\min} + \epsilon, \theta_{\max}]$ be a compact interval such that
\begin{equation}
\tilde{A}(\theta_{\min} + \epsilon) < A_{\mathrm{grid}} \quad \text{and} \quad \tilde{A}(\theta_{\max}) < A_{\mathrm{grid}}.
\end{equation}
Let $L$ be a Lipschitz constant for $\tilde{A}$ on $[\theta_{\min} + \epsilon, \theta_{\max}]$.
Then
\begin{equation}
A_{\mathrm{grid}} \leq A_{\mathrm{cl}}^*(p) \leq A_{\mathrm{grid}} + \frac{L\delta}{2}.
\end{equation}
\end{proposition}

\begin{proof}
The lower bound $A_{\mathrm{grid}} \leq A_{\mathrm{cl}}^*(p)$ is immediate since $A_{\mathrm{grid}}$ is achieved by a feasible strategy.

For the upper bound, the boundary conditions $\tilde{A}(\theta_{\min} + \epsilon) < A_{\mathrm{grid}}$ and $\tilde{A}(\theta_{\max}) < A_{\mathrm{grid}}$, combined with continuity of $\tilde{A}$ and the finite boundary limits from Lemma~\ref{lem:boundary}, ensure that the supremum of $\tilde{A}$ is attained on $[\theta_{\min} + \epsilon, \theta_{\max}]$.

For any $\theta_A$ in this interval, there exists a grid point $\theta_A^{(k)}$ with $|\theta_A - \theta_A^{(k)}| \leq \delta/2$.
By the Lipschitz condition:
\begin{equation}
\tilde{A}(\theta_A) \leq \tilde{A}(\theta_A^{(k)}) + L \cdot \frac{\delta}{2} \leq A_{\mathrm{grid}} + \frac{L\delta}{2}.
\end{equation}
Taking the supremum over $\theta_A$ yields the result.
\end{proof}

\begin{remark}[Practical computation]
In practice, we:
\begin{enumerate}
\item Choose $\epsilon$ small enough and $\theta_{\max}$ large enough that the boundary conditions in Proposition~\ref{prop:certified-bound} are satisfied (verified numerically).
\item Compute $A_{\mathrm{grid}}$ on a grid with spacing $\delta$.
\item Compute the Lipschitz constant $L$ numerically by evaluating $|d\tilde{A}/d\theta_A|$ on a finer grid.
\item Report the certified upper bound $A_{\mathrm{cl}}^*(p) \leq A_{\mathrm{grid}} + L\delta/2$.
\end{enumerate}
For demonstrating quantum advantage, it suffices to show $A_{\mathrm{qu}}(p) > A_{\mathrm{grid}} + L\delta/2$.
\end{remark}

\subsubsection{Concavification}

By Theorem~\ref{thm:concavification}, the correct classical benchmark is $A_{\mathrm{cl,SR}}^*(p) = \mathrm{conc}(A_{\mathrm{cl}}^*)(p)$.
The concave envelope is computed via the upper hull algorithm in $O(n)$ time for sorted input.

\subsection{Quantum Lower Bound via Polynomial Measurement Angles}

For quantum strategies, we consider players sharing a maximally entangled two-qubit state and performing local measurements parameterized by angles $\theta_A(x_1)$ and $\theta_B(x_2)$.
We are free to choose any Bell state and measurement angle convention; our numerical implementation uses a convention such that the resulting correlation is
\begin{equation}
    \mathbb{E}[o_A(X_1) o_B(X_2) \mid X_1 = x_1, X_2 = x_2] = \cos\bigl(2(\theta_A(x_1) - \theta_B(x_2))\bigr).
\end{equation}

We parameterize the measurement angles as polynomials of degree $n$:
\begin{equation}
    \theta_A(x) = \sum_{k=0}^{n} a_k x^k, \qquad \theta_B(x) = \sum_{k=0}^{n} b_k x^k.
\end{equation}
The game payoff and splitting probability become
\begin{align}
    A(a, b) &= -\mathbb{E}\bigl[\cos\bigl(2(\theta_A(X_1) - \theta_B(X_2))\bigr) \cdot w(X_1, X_2)\bigr], \\
    p(a, b) &= \mathbb{E}\left[\frac{1 - \cos\bigl(2(\theta_A(X_1) - \theta_B(X_2))\bigr)}{2}\right].
\end{align}

For computational efficiency, we evaluate these expectations using Gauss--Laguerre quadrature.
For $X \sim \mathrm{Exp}(\mu)$, we use the substitution $y = \mu x$ to obtain quadrature points $\{x_i\}$ and weights $\{w_i\}$ such that
\begin{equation}
    \mathbb{E}[f(X)] \approx \sum_{i=1}^{m} w_i f(x_i).
\end{equation}
We precompute the weight matrices
\begin{equation}
    W_{ij} = w(x_i, x_j) \cdot w_i w_j, \qquad \pi_{ij} = w_i w_j,
\end{equation}
enabling rapid evaluation of the payoff and constraint for any polynomial coefficients via array operations.

We maximize the payoff subject to the splitting probability constraint using sequential least-squares programming (SLSQP), with multiple random restarts to avoid local optima.
The result is a lower bound on the quantum value, as we are exhibiting a specific (polynomial) quantum strategy.

\subsection{Threshold Payoff Computation}

The $w$-threshold payoff $A^*(p)$ corresponds to the optimal policy established in Appendix~\ref{sec:delta_threshold}: split when $w(X_1, X_2) > \tau_p$ and bunch otherwise, where $\tau_p$ is chosen such that $\Pr[w(X_1, X_2) > \tau_p] = p$.

We compute $A^*(p)$ via Monte Carlo estimation:
\begin{enumerate}
    \item Sample $N$ independent pairs $(X_1^{(i)}, X_2^{(i)}) \sim \mathrm{Exp}(\mu)^{\otimes 2}$.
    \item Compute $w^{(i)} = w(X_1^{(i)}, X_2^{(i)})$ for each sample.
    \item Determine $\tau_p$ as the $(1-p)$-quantile of $\{w^{(i)}\}$.
    \item Compute the $w$-threshold decision $\sigma^{*(i)} = \mathrm{sign}(\tau_p - w^{(i)})$.
    \item Estimate the payoff as $A^*(p) \approx -\frac{1}{N}\sum_{i=1}^{N} \sigma^{*(i)} w^{(i)}$.
\end{enumerate}

\subsection{Throughput Model}

The throughput $\mathcal{T}(p)$ as a function of splitting probability is determined by the baseline task model described in Appendix~\ref{sec:baseline_throughput}.
For the numerical results presented in the main text, we use the exponential saturation warm-up model from Section~\ref{sec:warmup_example}, giving
\begin{equation}
\mathcal{T}(p) = \phi_{\max}(1-\rho) \cdot \frac{2\alpha}{\lambda(1+p) + 2\alpha}
\end{equation}
with $\phi_{\max} = 1$ and $\alpha = 0.5$.

\subsection{Numerical Parameters}
Table~\ref{tab:numerical-params} summarizes the numerical parameters used to generate the results in the main text.
\begin{table}[h]
\centering
\begin{tabular}{lll}
\hline
\textbf{Parameter} & \textbf{Value} & \textbf{Description} \\
\hline
$\lambda$ & 0.8 & Customer pair arrival rate \\
$\mu$ & 1.0 & Service rate \\
$\alpha$ & 0.5 & Warm-up rate \\
$\phi_{\max}$ & 1.0 & Steady-state productivity \\
\hline
\multicolumn{3}{l}{\textit{Classical optimization}} \\
Grid points & 500 & Number of $\theta_A$ grid points \\
$\theta_{\max}$ & 12.0 & Maximum threshold considered \\
\hline
\multicolumn{3}{l}{\textit{Quantum optimization}} \\
Polynomial degree & 2 & Degree of measurement angle polynomials \\
Quadrature points & 60 & Gauss--Laguerre quadrature points \\
Restarts & 20 & Random restarts for optimization \\
Seed & 1 & Random seed for reproducibility \\
\hline
\multicolumn{3}{l}{\textit{$w$-threshold computation}} \\
Monte Carlo samples & $10^5$ & Samples for payoff estimation \\
\hline
\end{tabular}
\caption{Numerical parameters used for computing classical and quantum performance.}
\label{tab:numerical-params}
\end{table}


\end{document}